\documentclass[11pt,noinfoline]{imsart}

\RequirePackage[OT1]{fontenc}
\RequirePackage{amsthm,amsmath}
\RequirePackage[round]{natbib}
\RequirePackage[colorlinks,citecolor=blue,urlcolor=blue]{hyperref}
\usepackage{ulem}
\usepackage{verbatim}
\usepackage{stmaryrd}
\usepackage{fullpage}
\usepackage{bbm}

\usepackage{layout}

\usepackage{float}
\usepackage{dsfont}
\usepackage[mathscr]{eucal}
\usepackage[toc,page]{appendix}
\usepackage{mathrsfs}
\usepackage{color}
\usepackage{pifont}
\usepackage{bm}
\usepackage{latexsym}
\usepackage{amsfonts}
\usepackage{amssymb}
\usepackage{epsfig}
\usepackage{graphicx}
\usepackage{multirow}
\usepackage{tikz}

\usepackage{caption}
\usepackage{subcaption}

\newtheorem{theorem}{Theorem}[section]
\newtheorem{corollary}[theorem]{Corollary}
\newtheorem{lemma}[theorem]{Lemma}

\newtheorem{remark}[theorem]{Remark}

\newtheorem{assumption}[theorem]{Assumption}

\usepackage{verbatim}

\usepackage{tikz}
\usetikzlibrary{fit,positioning,arrows,automata,calc}
\tikzset{
main/.style={circle, minimum size = 5mm, thick, draw =black!80, node distance = 10mm},
connect/.style={-latex, thick},
box/.style={rectangle, draw=black!100}
}

\usepackage{mathtools}
\usepackage{tabu}

\hypersetup{
colorlinks=true,       % false: boxed links; true: colored links
linkcolor=blue,        % color of internal links
citecolor=blue,        % color of links to bibliography
filecolor=magenta,     % color of file links
urlcolor=blue
}
\newcommand*\diff{\mathop{}\!\mathrm{d}}

\newcommand{\floor}[1]{\lfloor {#1} \rfloor}

\newcommand{\W}{\mathcal{W}}
\newcommand{\E}{\mathbb{E}}
\newcommand{\p}{\mathbb{P}}

\newcommand\norm[1]{\left\lVert#1\right\rVert}
\newcommand{\St}{\mathbf{S}}
\newcommand{\at}{\boldsymbol{\alpha}}
\newcommand{\id}{\mathrm{id}}

\definecolor{Green}{HTML}{59A14F}
\definecolor{Red}{HTML}{E15759}
\definecolor{Orange}{HTML}{F28E2B}
\definecolor{Blue}{HTML}{4E79A7}

\begin{document}

\begin{frontmatter}

    \title{{\large On Distributional Autoregression and Iterated Transportation}}

    \begin{aug}
    \author{\fnms{Laya} \snm{Ghodrati}\ead[label=e1]{laya.ghodrati@epfl.ch}} \and
    \author{\fnms{Victor M.} \snm{Panaretos}\ead[label=e2]{victor.panaretos@epfl.ch}}

    \runauthor{L. Ghodrati \& V.M. Panaretos}
   
    \affiliation{Ecole Polytechnique F\'ed\'erale de Lausanne}
   
    \address{Institut de Math\'ematiques\\
    Ecole Polytechnique F\'ed\'erale de Lausanne}
   
    \end{aug}

    \begin{abstract}
    We consider the problem of defining and fitting models of autoregressive time series of probability distributions on a compact interval of $\mathbb{R}$. An order-$1$ autoregressive model in this context is to be understood as a Markov chain, where one specifies a certain structure (regression) for the one-step conditional Fr\'echet mean with respect to a natural probability metric. We construct and explore different models based on iterated random function systems of optimal transport maps. While the properties and interpretation of these models depend on how they relate to the iterated transport system, they can all be analyzed theoretically in a unified way. We present such a theoretical analysis, including convergence rates, and illustrate our methodology using real and simulated data. Our approach generalises or extends certain existing models of transportation-based regression and autoregression, and in doing so also provides some additional insights on existing models.
    \end{abstract}

   
    \end{frontmatter}

	\doublespacing
    
    \section{Introduction}\label{intro}
   
In distributional regression, one aims to describe/estimate the relationship between a response distribution $\nu$, and a covariate distribution $\mu$, viewed as random measures. This is to be done on the basis of an i.i.d. sample of random pairs $\{(\mu_i,\nu_i)\}_{i=1}^{n}$. The relationship is modelled \textit{globally}, in that the complete distributions (seen as random elements of a suitable function space) are being related. In this sense, such models are useful in contexts where one can access samples from each law marginally, rather than in pairs (also known as uncoupled regression data). This can be due to data collection limitations, or simply because there is no natural coupling.

In light of this global perspective, distributional regression falls under the label of functional regression -- where one random function is to be regressed on another \citep{morris2015functional}. However, the non-linear nature of probability distributions distinguishes distributional regression from typical functional regression. In usual functional regression, one can model the regression via the usual (Bochner) conditional expectation and bounded linear transformations on Hilbert spaces \citep{hsing2015theoretical}. But these concepts do not readily apply in distributional regression, where one is confronted with the challenges of geometrical data analysis \citep{petersen2022modeling,patrangenaru2015nonparametric}. Early approaches to distributional regression circumvented this issue by imbedding the distributions in Hilbert space via suitable transformations \citep{kneip2001inference,delicado2011dimensionality,petersen2016functional,kokoszka2019forecasting}. More recently, attention has focussed on directly modeling random distributions as random elements of the Wasserstein space, a geodesic metric space related to optimal transport, increasingly seen as a canonical setting for distributional statistics \citep{panaretos2020invitation}. In this context, Bochner integrals are replaced by Fr\'echet means \citep{panaretos2016amplitude,bigot2018upper,zemel2019frechet}, and what remains is the choice of regressor function, i.e. the specification of a relationship linking the conditional Fr\'echet mean of the response to the covariate.

Two general strategies have arisen for this specification. The geometrical approach uses the fact that the Wasserstein space is locally Hilbert-like, and defines classical Hilbertian regression by lifting covariate and response on an appropriate tangent space (see \citet{chen2021wasserstein} and \citet{zhang2022wasserstein}). While this model has a natural mathematical interpretation, its statistical interpretation is somewhat contrived. The other strategy is to directly specify the regression transformation as an optimal transport map, exploiting convexity and shape constraints, rather than geometrical features (see \citet{ghodratidistribution}). This has the advantage of a clean interpretation and of avoiding ill-posedness issues.

Distributional autoregression is a natural next-step for distributional regression models -- indeed, it is arguably the setting where most distributional regression data sets arise. Rather than i.i.d. covariate/response distributions, one observes a dependent sequence of probability distributions $\{\mu_n\}_{n=1}^N$. When viewed as a Markov chain in the Wasserstein space, this sequence can be modeled autoregressively by specifying a relationship between the conditional Fr\'echet mean at time $n+1$ and the chain at time $n$. Once again, this can be done geometrically (as indeed was already explored in \citep{chen2021wasserstein} and  \citep{zhang2022wasserstein}), or by way of optimal transport maps, with similar advantages/disadvantages.

A first contribution based directly on transport maps was made in \citet{zhu2021autoregressive}, where random perturbations of the identity were iteratively contracted/composed to form a time-dependent sequence. This was subsequently used either as ``increments" between consecutive distributions or as ``deviations'' from the marginal Fr\'echet mean, to produce autoregressive models. Key in this approach was the use of iterated random function systems and a canny definition of a contraction operation on the space of transport maps, allowing to mimic the contractive effect of a correlation operator in usual autoregression.  \citet{jiang2022wasserstein} subsequently generalised this approach to autoregressive modeling to the case of vector-valued distributional chains, i.e. time-evolving vectors with distributions as coordinates.

A salient limitation of this approach is that the entire dynamics of the process reduce to a single scalar quantity $|\alpha|\leq 1$, regulating the ``strength" of the contraction.  While this resembles real-valued autoregressive processes, it is likely too rigid in a functional context (or even a multivariate context), and can have undesirable consequences when asserting stationarity (see Section \ref{related works} for a more extensive discussion). Ideally, a genuinely functional model would allow for a \textit{functional} specification of the dynamics, thus capable of expressing more complex dependencies. In response to this drawback, \citet{zhu2021autoregressive} also defined a model where the scalar contraction coefficient is replaced by a \textit{functional contraction coefficient}, contracting variably across the domain. This comes with the caveat of a more complicated theory, including cumbersome technical assumptions, as well as a more involved interpretation.

The purpose of this paper is to introduce and develop transportation-based autoregressive models with genuinely functional dynamics, yielding easily interpretable yet rich classes of distributional autoregressions. To do so, we extend to the autoregressive case the functional structure of \citet{ghodratidistribution}, where the regression operator is a monotone rearrangement, making use of the scalar ``contractive effect" introduced by \citet{zhu2021autoregressive} -- intuitively, we posit a model where the \textit{shape} of the dynamics is captured by a monotone map, modulated by a contractive parameter $\alpha$ regulating the degree of non-degeneracy of the model. In its simplest form, this approach can be interpreted as positing that
 $$\mu_{n+1}=\theta_n\#[\alpha\mu_n],  \quad n \in\mathbb{Z},$$
for i.i.d. random increasing maps $\theta_n$ with $\mathbb{E}[\theta_n(x)]=S(x)$; $S$ a deterministic monotone map; and $\mu_n\mapsto [\alpha\mu_n]$ a barycentric contraction operation, suitably defined at the level of quantile functions (see Equation \eqref{scalar_multiplication} for a precise definition). Intuitively, the model suggests that step $n+1$ in the chain is obtained by pushing forward the $n$th step (``shrunken" slightly to allow for temporal stationarity) via a random perturbation of the deterministic deformation $S$. This is a direct autoregressive extension of \citet{ghodratidistribution}, employing the contractive device of \citet{zhu2021autoregressive} to assure temporal stability in law. However, more modeling possibilities are available in our approach, and this is just the motivating one (see Section \ref{interpretations}).

The rest of the paper is organised as follows. After introducing some basic background and notation (Section \ref{Wasserstein}), we revisit the problem of defining iterated random function systems of increasing maps. In particular, Section \ref{iterations} presents a functional extension of the iterated system employed in \citet{zhu2021autoregressive}. This extension is then used in Section \ref{interpretations} in order to define three different possible notions of autoregression -- in each case, the iterated transport map system serves to model a different characteristic of the distributional time series (e.g. the increments, the quantiles, or the generalised quantiles). We compare the resulting models to existing approaches in Section \ref{related works} and determine conditions for stationarity in Section \ref{stationarity}. We then show in Section \ref{estimation} that all three models can be fitted and analysed using the same estimation theory -- albeit applied to optimal maps that represent a different characteristic in each case. In particular, we establish identifiability, consistency, and rates of convergence. Finally, the finite sample performance of our methodology is illustrated on some simulated and real data (Sections \ref{simulations} and \ref{real_data_analysis}). The proofs are collected in a separate Section, and we conclude with a discussion of some further possible generalisations.

    \section{Background on Optimal Transport and Some Notation}{\label{Wasserstein}}
    In this section, we provide some background on optimal transport and associated probability metrics. For more background see, e.g. \cite{panaretos2020invitation}. Let $\Omega = [\omega_0,\omega_1]$ be a closed interval of $\mathbb{R}$ and $\W_2(\Omega)$ be the set of Borel probability measures on $\Omega$, with finite second moment. The 2-Wasserstein distance $W$ between $\mu,\nu \in \W_2(\Omega)$ is defined by
    $$d^2_{\W}(\nu,\mu):=\underset{\gamma \in \Gamma(\nu,\mu)}{\inf} \int_{\Omega} |x-y|^2 \diff \gamma(x,y),$$
    where $\Gamma(\nu,\mu)$ is the set of couplings of $\mu$ and $\nu$, i.e. the set of Borel probability measures on $\Omega \times \Omega$ with marginals $\nu$ and $\mu$. It can be shown that $\W_2(\Omega)$ endowed with $d^2_{\W}$ is a metric space, which we simply call the Wasserstein space of distributions. A coupling $\gamma$ is deterministic if it is the joint distribution of $\{X, T(X)\}$ for some deterministic map $T: \Omega \to \Omega$, called an optimal transport map. In such a case, we write $\nu=T\#\mu$ and say that $T$ pushes $\mu$ forward to $\nu$, i.e. $\nu(B)=\mu\{T^{-1}(B)\}$ for any Borel set $B$.

    \begin{remark}
        Throughout the paper, we will focus on invertible maps (hence strictly increasing).
    \end{remark}

    When the source distribution $\mu$ is absolutely continuous with respect to the Lebesgue measure,  then the optimal plan is induced by a map $T$. When $d=1$, the map $T$ is a nondecreasing map and admits the explicit expression $T=F^{-1}_{\nu}\circ F_\mu$, where $F^{-1}_\nu$ is the quantile function of $\nu$, and $F_\mu$ is the cumulative distribution function of $\mu$. It follows immediately that the composition of two optimal maps results in another optimal map. In addition, we have the explicit expression
   
    \begin{equation}\label{w1d}
      d_{\W}^2(\mu,\nu)=\int_0^1\big|F^{-1}_{\mu}(p) - F^{-1}_{\nu}(p)\big|^2 \diff p.
    \end{equation}

    Finally, we will use the notation $a \lesssim b$ to indicate that there exists a positive constant $C$ for which $a\leq C b$ holds (bounded above up to a universal constant). We denote by $\norm{.}_p$ the usual $L^p$ norm of a function.
   
    \section{Autoregressive Models via Iterated Transportation}
    \subsection{Random Iterated Transport}\label{iterations}

    Our definition of autoregressive models for distributions will hinge on appropriately defined iterated random systems of transport maps (following the approach of \citet{zhu2021autoregressive}, to whom we compare below). This is a special case of a framework for studying questions about Markov chains via iterated random functions, going back to at least \citet{diaconis1999iterated}. They define an iterated random function system on a state space $\mathcal{T}$ as
    $$T_i=f(T_{i-1}\,;\,\theta_i)$$
    for a family of transformations $\{f(\,\cdot\,;\,\theta):\theta\in\Theta\}$ acting on $\mathcal{T}$,  and random elements $\theta_i$ in some parameter space $\Theta$,  independent of $T_i\in\mathcal{T}$. By suitable choice of the family $f(\,\cdot\,;\,\theta)$ and some distribution on $\Theta$ they show how a plethora of Markov chains can be cast in this light.
   
    In our case, the state space $\mathcal{T}$ will be the set of optimal transport maps
    $$\mathcal{T}:=\{T :\Omega \to \Omega| T(\omega_1) = \omega_1 , T(\omega_2)=\omega_2, T \text{ is strictly increasing and continuous} \},$$
    viewed as a closed and complete subset of the Lebesgue space $L^p(\Omega)$ equipped with the corresponding $p$-distance $\|\cdot\|_p$, for some $ 1\leq p < \infty$ (we will mostly focus on $p=2$). And, the question is how to define $f$ and $\theta_i$ to generate an iterated random system that is sufficiently rich to serve as a basis for interesting autoregressive models, yet remains tractable and admits a non-degenerate stationary solution. Naively, one might simply posit that $\Theta=\mathcal{T}$ and $f(T;\theta)=\theta\circ T$, as increasing maps form a transformation group under composition. However,  $f_\theta$ needs to be a contraction ``on average" (in a precise sense) for the \citet{diaconis1999iterated} results to be applicable.
   
    This motivates forms of $f$ that are ``contractive compositions". To this aim, given $|\alpha|\leq  1$, define the $\alpha$-contraction of an optimal transport map to be the operator $T\mapsto [\alpha T]$ defined pointwise via    
    \begin{equation}{\label{scalar_multiplication}}
        [\alpha T](x) = \begin{cases}
            x+\alpha(T(x)-x) & 0<\alpha \leq 1 \\
            x & \alpha=0 \\
            x+\alpha(x-T^{-1}(x)) & -1\leq\alpha <0.
         \end{cases}
    \end{equation}
   
    This definition is due to \citet{zhu2021autoregressive}, under slightly different terminology/notation, and mimics the operation of contracting an unconstrained function by a scalar, but conforming to the constraints elicited by working in $\mathcal{T}$. Notice that $T\mapsto [\alpha T]$ is indeed a contraction on $\mathcal{T}$ with respect to $L^1$ norm, with the identity as its fixed point -- any other fixed point must equal the identity almost everywhere by the Banach fixed-point theorem.
   
    Finally, given $|\alpha|<1$ and $\theta\in\mathcal{T}$ we can now make precise the notion of $f$ being a ``contractive composition" map by defining
    $$ f(T;\theta)=\theta\circ [\alpha T].$$
    To define an iterated random system, it suffices to put a probability distribution $Q$ on $\mathcal{T}$, and make i.i.d. draws $\theta_i\sim Q$ yielding
    \begin{equation}\label{general-iteration}T_i = f(T_{i-1};\theta_i).\end{equation}

    Our proposal is to draw i.i.d. elements of $\mathcal{T}$ with a specified expectation $S \in\mathcal{T}$, say $\theta_i = T_{\epsilon_i} \circ S,$ for $\{T_{\epsilon_i}\}_{i=1}^{N}$ a collection of independent and identically distributed random optimal maps satisfying $\E\{T_{\epsilon_i}(x)\}=x$ almost everywhere on $\Omega$. Explicitly, our iteration is now
    \begin{equation}\label{Model}
        T_i = f(T_{i-1}; \underset{\theta_i}{\underbrace{T_{\epsilon_i} \circ S}})=\underset{\theta_i}{\underbrace{T_{\epsilon_i} \circ S}} \circ [\alpha T_{i-1}].
    \end{equation}
     The degrees of freedom in this iteration are the choice of $S\in\mathcal{T}$ and $\alpha\in [-1,1]$. In a statistical setting, these would be the targets of estimation. This definition extends the iteration of \citet{zhu2021autoregressive} where $S$ was a priori fixed to be the identity. Our extension seems natural and conceptually straightforward: it iterates contracted composition with perturbations of an arbitrary element of the transformation group, rather than with perturbations of the neutral element.  Yet, it substantially complicates the subsequent probabilistic analysis and estimation theory. In exchange, we get a richer class of autoregressive models that exhibit advantages in the context of modeling and data analysis.  We elaborate on the relationship and the nature of the extension in a subsequent paragraph. We then show that the iteration admits a unique stationary solution (under some additional assumptions). First, though, we explore how such an iterated random system of optimal maps could be used as a basis for distributional autoregression.

    \subsection{Autoregressive Models}{\label{interpretations}}
   
    The main purpose of a random iteration \eqref{general-iteration} is the construction of a Markov chain model for a dependent sequence of probability distributions $\mu_i\in\W_2(\Omega)$, that will always be taken to possess a continuous cumulative distribution function. The models we seek are of autoregressive type, and so should ultimately be interpretable as a structural specification of the one-step conditional mean. Given stationary random sequence $\{T_i\}$ of optimal maps, there appear to be (at least) three different ways of doing so, by relating the $T_i$ to some suitable feature of $\{\mu_i\}$:
   
    \begin{enumerate}
    \medskip
    \item[(I)] Modeling the ``increments" $T_{\mu_{i-1}}^{\mu_i}:=F^{-1}_{\mu_i}\circ F_{\mu_{i-1}}$ as being equal to $T_i$ (we call these increments, as $T_{\mu_{i-1}}^{\mu_i}$ is the optimal map pushing $\mu_{i-1}$ forward to $\mu_i$), or equivalently modeling the quantiles as
    $$ F^{-1}_{\mu_i}:= T_i\circ F^{-1}_{\mu_{i-1}}.$$
     When $\{T_i\}$ is stationary, this yields a process with stationary increments, but the process could be non-stationary (if so, it's interesting to understand if there is ``drift"). This chain corresponds to specifying that the
    (usual) conditional expectation of $F^{-1}_{\mu_i}$ given $F^{-1}_{\mu_{i-1}}$ as
    $$\mathbb{E}[F^{-1}_{\mu_i}|F^{-1}_{\mu_{i-1}}]=\mathbb{E}\{T_i\}\circ F^{-1}_{\mu_{i-1}}=\mathbb{E}\{f(T_{i-1};\theta_i)\}\circ F^{-1}_{\mu_{i-1}}=\mathbb{E}\{\theta_i\circ[\alpha T_{i-1}]\}\circ F^{-1}_{\mu_{i-1}}.$$

  The precise form of $\mathbb{E}[T_i]$ will depend on the stationary solution of $T_i=f(T_{i-1};\theta_i)$. \\

    \medskip
    \item[(UQ)] Modeling the (uniform) quantiles $F_{\mu_i}^{-1}$ as being equal to $T_i$,
    $$F_{\mu_i}^{-1}:=T_i.$$
    This automatically yields a stationary process when $\{T_i\}$ is stationary, directly interpretable at the level of quantiles, and corresponds to specifying the (usual) conditional expectation of $F^{-1}_{\mu_i}$ given $F^{-1}_{\mu_{i-1}}$ as
    $$\mathbb{E}[F^{-1}_{\mu_i}|F^{-1}_{\mu_{i-1}}]=(\mathbb{E}\theta_i)\circ [\alpha F^{-1}_{\mu_{i-1}}]=S\circ [\alpha F^{-1}_{\mu_{i-1}}]=f(F^{-1}_{\mu_{i-1}};S).$$
   This model corresponds to an autoregressive extension of the model in \citet{ghodratidistribution}.

    \medskip
    \item[(GQ)] Modeling the generalised quantiles \citep{chernozhukov2017monge} or $\mu$-quantiles $F^{-1}_{\mu_i}\circ F_\mu$  with respect to some measure $\mu$ as being equal to $T_{i}$. This also immediately yields stationarity and (under regularity conditions) is equivalent to stating $\mu_{i}=T_{i}\#\mu$, in effect modeling the $\mu_i$ as serially dependent ``perturbations" of a fixed $\mu$. This corresponds to specifying the (usual) conditional expectation of $F^{-1}_{\mu_i}$ given $F^{-1}_{\mu_{i-1}}$ as
    $$\mathbb{E}[F^{-1}_{\mu_i}|F^{-1}_{\mu_{i-1}}]=(\mathbb{E}\theta_i)\circ[\alpha [F^{-1}_{\mu_{i-1}}\circ F_\mu]]=S\circ [\alpha [F^{-1}_{\mu_{i-1}}\circ F_\mu]]=f(F^{-1}_{\mu_{i-1}}\circ F_\mu;S).$$
    \end{enumerate}

    Note that setting $\alpha = 1$ in (UQ) yields the same model as setting $\alpha = 0$ in (I), interpretable as a random walk, and this we shall revisit. In Section \ref{real_data_analysis} we will focus on (UQ) and (I) to model sequential distributional data and discuss the merits/drawbacks of each approach. Model (GQ) can actually be seen to be a variant of the model (UQ) albeit under a modification of the definition of the contraction operator itself  -- see Section \eqref{sec:gq-model}, and especially Remark \eqref{rem:gq-model} for an equivalent characterization of the model (GQ)

    \color{black}

    \subsection{Comparison with Related Work}{\label{related works}}
   
    Our iteration \eqref{Model} represents a  generalization of the iteration in
    \citet{zhu2021autoregressive}, by combining their notion of $\alpha$-contraction (which they call \textit{distributional scalar multiplication}), with the functional structure of the model in \citet{ghodratidistribution}. Specifically, \citet{zhu2021autoregressive} considered autoregressive models for distributional time series, based on the iterative system of optimal transport maps
    \begin{equation}{\label{model_muller}}
        T_i = T_{\epsilon_i} \circ [\alpha T_{i-1}].
    \end{equation}
    That this is a special case of our system \eqref{Model} when $S$ is fixed to be the identity map $\id(x)=x$. Their clever $\alpha$-contraction, combined with classical results on iterated random function theory, allows one to deduce the existence of a unique stationary solution to the iteration \eqref{model_muller} thanks to the contracting effect of $\alpha$ for $-1<\alpha<1$ (and some additional technical assumptions).
       
    \noindent However, basing a distributional autoregressive model on this system is restrictive in two important ways:
    \begin{enumerate}
    \item As a stochastic model, the system \eqref{model_muller} is parametric and univariate: the only unknown is the scalar coefficient $\alpha\in (-1,1)$. Correspondingly, when basing our model on that iteration (with any of the three interpretations specified in the previous section), the temporal dependence of $\mu_{i}$ on $\mu_{i-1}$ will be completely specified up to an unknown scalar parameter. This is reminiscent of autoregressive models on the real line but is arguably overly restrictive in a functional data analysis (or even multivariate analysis) setting, where the temporal dependence is very likely more complex. A genuinely functional model would replace the scalar coefficient with a suitable \textit{functional coefficient}, e.g. a non-linear operator.
     \\
     
    \item If a stationary solution to system \eqref{model_muller} exists,  then it must satisfy  $\E(T_i) = \id$. To see this, recall the definition of the scalar multiplication \eqref{scalar_multiplication} and observe that
    $$
            \E[T_{i}] = \E[T_{i+1}]
            = \E[\E[T_{i+1}|T_{i}]]
            = \E[\alpha T_{i}].
    $$
    This is consequential if using the sequence $T_i$ to induce a distributional time series $\{\mu_i\}$. In the (I) model, where $T_i$ models the increments between consecutive $\mu_i$, this implies that the conditional Fr\'echet mean (in the Wasserstein metric) of $\mu_i$ given $\mu_{i-1}$ is exactly equal to $\mu_{i-1}$, a sort of `Fr\'echet martingale'. Effectively this trivializes the regressor relationship to be an identity -- there is no modeling flexibility for the conditional mean, only the conditional variance (via $\alpha$). In the (UQ) model, where $T_i\equiv F^{-1}_{\mu_i}$ is taken as the quantile function of $\mu_i$, the fact that $\E(T_i) = \id$ implies that the distributional autoregression model can only admit the uniform distribution as its Fr\'echet mean (with respect to the Wasserstein metric). There is no flexibility in the modeling of the marginal mean.  \end{enumerate}
    By contrast, models based on our system \eqref{Model} are genuinely functional, since on account of the unknown transport map $S$. Furthermore, our model can accommodate \textit{any} distribution as its Fr\'echet mean: given any optimal map $T\in \mathcal{T}$,  there exist $S$ and $\alpha$ such that $\E(T_i) = T$.

    The optimal map interpretation of our system \eqref{Model} is an auto-regressive modification of the distributional optimal transport regression model of \citet{ghodratidistribution}.  \citet{ghodratidistribution} define the regression model
    $$\nu_{i}=T_{\epsilon_i}\#(S\#\mu_i),  \quad  \{\mu_i,\nu_i\}_{i=1}^N,$$
    where $S:\Omega \to \mathbb{R}$ is an unknown optimal map and $\{T_{\epsilon_i}\}_{i=1}^{N}$ is a collection of independent and identically distributed random optimal maps satisfying $\E\{T_{\epsilon_i}(x)\}=x$ almost everywhere on $\Omega$.
    By direct analogy, an autoregressive model (optimal map interpretation) for a time series of distributions $\{\mu_i\}$ would be defined as
    \begin{equation}{\label{model}}
    \mu_{i}=T_{\epsilon_i}\#(S\#\mu_{i-1}),
    \end{equation}
    which is equivalent to model \eqref{Model} when $\alpha=0$ and when we interpret $T_i$ such that $\mu_i = T_i \# \mu_{i-1}$, i.e. the optimal map interpretation. If we take the quantile interpretation, the two models are again related for $\alpha = 1$ since model \eqref{Model} is equivalent to
    $$F^{-1}_i = T_{\epsilon_i} \circ S \circ F^{-1}_{i-1}.$$
   
    However, assuming the noise maps $T_{\epsilon_i}$ are close to identity, one observes that the series of CDFs $F^{-1}_i$ would stabilize around a step function where the position of the jumps coincide with fixed points of the map $S$, and therefore the distribution $\mu_i$ would oscillate around a mixture of Dirac measures. This is where we combine the functional structure of \citet{ghodratidistribution} with the scalar ``contractive effect" introduced by \citet{zhu2021autoregressive} -- intuitively, the magnitude of $\alpha$ regulates the non-degeneracy of the model. The next Section demonstrates that this combined extension does indeed yield a unique stationary solution.
   
    \subsection{Existence of Unique Stationary Solution}\label{stationarity}
    We now turn to establish the existence of a unique stationary solution for the system \eqref{Model}. We will use the results of \citet{wu2004limit}, extending to our iteration \eqref{Model} the steps follows by \citet{zhu2021autoregressive} in the context of iteration \eqref{model_muller}.
    Let $\{T_{\epsilon_i}\}_{i=1}^{N}$ be a collection of independent and identically distributed random optimal maps satisfying $\E\{T_{\epsilon_i}(x)\}=x$ almost everywhere on $\Omega$.     Define $\Phi_i,\tilde{\Phi}_{i,m}:\mathcal{T}\to\mathcal{T}$ by
    \begin{equation}{\label{iterated_system}}
    \begin{split}
    &\Phi_i(T) = f(T;T_{\epsilon_i}\circ S)=T_{\epsilon_i} \circ S \circ [\alpha T] \\
    &\tilde{\Phi}_{i,m}(T) = \Phi_i\circ \Phi_{i-1}\circ\cdots\circ \Phi_{i-m+1}(T).
    \end{split}
    \end{equation}

   The following assumption stipulates
    \begin{assumption}(Moment Contracting Condition \citep{wu2004limit})\label{moment_contracting_assumption}
        Suppose there exists $\eta>0, Q_0\in \mathcal{T}, C>0$ and $r\in(0,1)$ such that
        \begin{equation}{\label{moment_contracting_eq}}
            \E\norm{\tilde{\Phi}_{i,m}(Q_0)-\tilde{\Phi}_{i,m}(T)}_2^\eta\leq Cr^m \norm{Q_0-T}_2^\eta
        \end{equation}
        holds for all $i\in \mathbb{Z}$, $m\in \mathbb{N}$ and all $T\in \mathcal{T}$.
    \end{assumption}
   
    \begin{lemma}{\label{stationary_solution}}
        Assume the parameters of the model \eqref{Model} satisfy the Assumption \ref{moment_contracting_assumption}. Then for all $T\in \mathcal{T}$, $\tilde{T}_i:=\lim_{m\to\infty} \tilde{\Phi}_{i,m}(T)\in \mathcal{T}$ exists almost surely and does not depend on $T$. In addition, $\tilde{T}_i$ is a stationary solution to the following system of stochastic transport equations:
        $$T_i = T_{\epsilon_i} \circ S \circ [\alpha T_{i-1}], \quad i \in \mathbb{Z},$$
        and is unique almost surely.
    \end{lemma}

\begin{remark}\label{sufficient_condition_existence}
     \citet{zhu2021autoregressive} proposed a specific parameter condition for their model that ensures Assumption \ref{moment_contracting_assumption} is satisfied. We provide a similar sufficient condition for the parameters of Model \eqref{Model} that also guarantees the satisfaction of Assumption \ref{moment_contracting_assumption}. Let $L_\epsilon$ be constant such that $\E|T_\epsilon(x)-T_\epsilon(y)|^2\leq L_\epsilon^2 |x-y|^2$. Assuming $\alpha\geq 0$, if $|S(x)-S(y)|\leq L_S|x-y|$ and $\alpha L_S L_{\epsilon}< 1$, then Model \eqref{Model} satisfies Assumption \ref{moment_contracting_assumption} with $\eta=2$ and $r = \sqrt{\alpha L_S L_\epsilon}$. Similarly, if $\alpha<0$, suppose the aforementioned conditions are met and define $\mathcal{T}_{l,u} = \{ T\in \mathcal{T}:  0<L_l\leq T' \leq L_u <\infty\}$ and assume $\{T_i\} \subset \mathcal{T}_{l,u} \subset \mathcal{T}$ (see Lemma \ref{norm_inverse}). Then Model \eqref{Model} also satisfies Assumption \ref{moment_contracting_assumption} with $\eta=2$ and $r = \sqrt{\alpha L_S L_\epsilon}$.
\end{remark}

\subsection{Estimation and Statistical Analysis}\label{estimation}

We consider a time series of continuous distributions $\mu_i\in\W_2(\Omega)$ and corresponding time series $T_i\in\mathcal{T}$, which are related by one of the models from section 3.2. Although the methods to obtain $T_i$ may differ for each model, we can always obtain $T_i$ by observing $\mu_i$. Our analysis is thus applicable to all three models studied, but in each different model, the $T_i$ will represent a different feature of the distributional time series. For the remainder of our analysis, we assume that $T_i$ is a (the) stationary solution obtained from system \eqref{Model}.

As discussed in Section \ref{related works}, when $S$ is fixed a priori to be the identity, our iteration \eqref{Model} will reduce to that of \citet{zhu2021autoregressive}. In this simplified setting, \citet{zhu2021autoregressive} use the fact that $\alpha$ is the minimizer of $\E\norm{T_{i+1}-[\alpha T_i]}_2^2$ to  obtain a closed form expression for $\alpha$ as
$$
    \frac{\displaystyle\int_\Omega \E[(T_{i+1}(x)-x)(T_i(x)-x)]\diff x}{\displaystyle\int_\Omega \E[(T_{i}(x)-x)^2]\diff x}
    $$
    when $\alpha\in[0,1)$ or
        $$\frac{\displaystyle\int_\Omega \E[(T_{i+1}(x)-x)(x-T^{-1}_i(x))]\diff x}{\displaystyle\int_\Omega \E[(x-T^{-1}_{i}(x))^2]\diff x}
$$
when $\alpha\in (-1,0)$. These show that $\alpha$ can be interpreted as the autocorrelation coefficient, and can be estimated by its empirical version, which allows for a straightforward path to consistency and parametric rates of convergence.

However, our more general iteration \eqref{Model}, involves an arbitrary non-decreasing map $S$ that also needs to be estimated. Consequently, not only are there no closed forms for the estimands $(\alpha,S)$ but the estimation problem becomes distinctly non-linear.

To motivate our estimators, we note that if $S$ were known, then $\alpha$ could be estimated by non-linear least squares, as the minimiser of $ \frac{1}{N} \sum_{i=1}^N\norm{S \circ [\alpha T_{i-1}]-T_i}^2_2$. On the other hand, if $\alpha$ were known, then a natural candidate to estimate $S$ would be the ergodic average
$$S_{N,\alpha}:=\frac{1}{N} \sum_{j=1}^N T_j \circ [\alpha T_{j-1}]^{-1}.$$
This is because the definition of the iteration $T_j=f(T_{j-1};T_{\epsilon_i}\circ S)=T_{\epsilon_i}\circ S\circ [\alpha T_{j-1}]$, combined with the assumption that $\E[T_{\epsilon_j}(x)]=x$, yields that
$$\E\{T_j\circ[\alpha T_{j-1}]^{-1}\}=\E\{T_{\epsilon_j}\circ S\}=S.$$  
Since $S_{N,\alpha}$ is available in closed form for any choice of $\alpha$, this suggests plugging the expression $S_{N,\alpha}$ for $S$ into the sum of squares, to obtain an objective that depends only on $\alpha$. Minimising the said objective over $\alpha$ one obtains an estimator $\hat\alpha$, which automatically induces an estimator of $S$ in the form of $S_{N,\hat\alpha}$.

\medskip
\noindent Formally,  we define the estimators $(\hat{\alpha}_N,S_{N,\hat{\alpha}_N})$ of $(\alpha,S)$ as follows:

\begin{equation}\label{estimator}
        \hat{\alpha}_N \coloneqq \arg\min_\alpha M_N(\alpha),
\end{equation}
where
\begin{equation}\label{estimator-components}
    \begin{split}
            &M_N(\alpha)\coloneqq \frac{1}{N} \sum_{i=1}^N g_\alpha(T_{i-1},T_i,S_{N,\alpha})\\
            &g_\alpha(T_{i-1},T_i,S) \coloneqq  \norm{S \circ [\alpha T_{i-1}]-T_i}^2_2\\
        &S_{N,\alpha} \coloneqq \frac{1}{N} \sum_{j=1}^N T_j \circ [\alpha T_{j-1}]^{-1}.
    \end{split}
\end{equation}

To analyse the behaviour of our estimators, we also define the following population quantities:

\begin{equation}
    \begin{split}
        &S_\alpha \coloneqq \E[T_j \circ [\alpha T_{j-1}]^{-1}]\\
        &M(\alpha) \coloneqq \E g_\alpha(T_{i-1},T_i,S_\alpha).
    \end{split}
\end{equation}
The left-hand sides do not depend on $j$ due to stationarity, which will be assumed throughout.

\medskip
For the sake of clarity, we will henceforth denote the true parameters of the model using boldface fonts, namely as $(\at,\St)$.
\begin{theorem}{\label{S_alpha}}
    If the true parameters of the model are $(\at,\St)$, then $S_{\at}=\St$.
\end{theorem}
\begin{proof}
    For the true $\at$, we have
    \begin{equation}
        \begin{split}
            S_{\at} &= \E[T_j \circ [\at T_{j-1}]^{-1}] \\
            &=  \E[T_{\epsilon_j}\circ \St\circ [\at T_{j-1}] \circ [\at T_{j-1}]^{-1}] \\
            &= \E[T_{\epsilon_j} \circ \St] = \St
        \end{split}
    \end{equation}
\end{proof}

We show the consistency of the estimators $(\hat{\alpha}_N,S_{N,\hat{\alpha}_N})$ in the following 4 steps corresponding to the lemmas \ref{minimizer_expectation}, \ref{CLT_maps}, \ref{CLT_g} and Theorem \ref{consistency} respectively:
\begin{itemize}
    \item $\at$ is the unique minimizer of $M(\alpha)$.
    \item $S_{N,\alpha}$ converges uniformly (with respect to $\alpha$) in probability to $S_\alpha$ in $L^2$.
    \item $M_N(\alpha)$ converges uniformly in probability to $M(\alpha)$.
    \item we conclude the consistency (and identifiability) using the M-estimation theory.
\end{itemize}

\begin{lemma}{\label{minimizer_expectation}}(Unique Minimizer of $M(\alpha)$)
    For any $\alpha\neq \at$ we have
    $$M(\at)=\E g_{\at}(T_{i-1},T_i,S_{\at}) < \E g_\alpha(T_{i-1},T_i,S_\alpha)=M(\alpha),$$
    where $\at$ is the true $\alpha$.
\end{lemma}

Now we show that $S_{N,\alpha}$ converges to $S_\alpha$ in probability for any $\alpha$ and also prove a central limit theorem (CLT) for $S_{N,\alpha}$.

If $\alpha = \at$, then it is straightforward to argue that $S_{N,\at}$ converges to $S_{\at}$: first note that for any $x\in[0,1]$, the strong law of large numbers yields that

$$S_{N,\at}=\frac{1}{N} \sum_{j=1}^N T_j \circ [\at T_{j-1}]^{-1} =\frac{1}{N}  \sum_{j=1}^N T_{\epsilon_j} \circ \St \to \E{[T_{\epsilon_j} \circ \St]}=\St.$$
Therefore the terms in the expression are independent and identically distributed with mean $\St$. From Theorem \ref{S_alpha}, we know that the true $\St=S_{\at}$. Therefore in this case that $\alpha=\at$, $S_{N,\alpha}$ converges in probability to $S_{\at}=\St$. However, in general, when $\alpha \neq \at$ the terms $T_j \circ [\alpha T_{j-1}]^{-1}$ are not independent for different $j$. Therefore, we first show that since $\{T_j\}$ satisfies the moment generating condition, we can quantify the dependency between the terms in the sequence $T_j \circ [\alpha T_{j-1}]^{-1}$ and apply CLT methods developed for functional time series.

\begin{lemma}{\label{m_dependent}}
    A sequence $\{T_n\}_{i=-\infty}^\infty$ that satisfies the geometric moment contracting condition (\ref{moment_contracting_assumption}) for $\eta\geq 2$, also satisfies the conditions (1.1),(1.2),(2.1) and (2.2) of \citet{horvath2013estimation}. Namely, assume $$T_n = f(\epsilon_n,\epsilon_{n-1},\cdots),$$
    where $\{\epsilon'_i\}$ is an independent copy of $\{\epsilon_i\}$ defined in the same probability space. Then, letting

    \begin{equation}
        T_{n,m}' = f(\epsilon_n,\epsilon_{n-1},\cdots,\epsilon_{n-m+1},\epsilon'_{n-m},\cdots),
    \end{equation}
    for any $0<\delta<1$ we have
    \begin{equation}{\label{m_dependent_approximation}}
        \sum_{m=1}^\infty (\E\norm{T_n-T'_{n,m}}^{2}_2)^{1/2} <\infty.
    \end{equation}

\end{lemma}

\begin{lemma}(Central limit for $S_{N,\alpha}$)\label{CLT_maps}
    Suppose the parameters of the iteration \eqref{Model} satisfy the Assumption \ref{moment_contracting_assumption}. Then for any $\alpha$, there is a Gaussian process $\Gamma_\alpha$ such that
    $$\sqrt{N} (S_{N,\alpha}-S_\alpha)  \overset{d}{\to} \Gamma_\alpha,  \quad \text{in } L^2.$$
    Also,
   $$\sup_\alpha \norm{S_{N,\alpha}-S_\alpha}_2 = o_\p(1)$$
\end{lemma}

\begin{lemma}\label{CLT_g}  Suppose the parameters of the iteration \eqref{Model} satisfy the Assumption \ref{moment_contracting_assumption}.
Then for any $\alpha$, there is a $\sigma_\alpha\geq 0$ such that
$$\sqrt{N} [M_N(\alpha)-M(\alpha)]\to N(0,\sigma^2_\alpha).$$
Moreover,
$$\sup_\alpha |M_N(\alpha)-M(\alpha)|=o_\p(1).$$

\end{lemma}

\begin{theorem}(Identifiability and Consistency)\label{consistency}
    Under Assumption \ref{moment_contracting_assumption}, the parameters of the iteration \eqref{Model} are identifiable and $(\hat{\alpha}_N,S_{N,\hat{\alpha}_N})$ are consistent estimators for $(\at,\St)$.

\end{theorem}

\begin{theorem}\label{rate}(Rate of Convergence)
Let $\mathcal{T}_{l,u} = \{ T\in \mathcal{T}:  0<L_l\leq T' \leq L_u <\infty\}$ and suppose $\{T_i\} \subset \mathcal{T}_{l,u} \subset \mathcal{T}$. Under Assumption \ref{moment_contracting_assumption} and twice differentiability of the $T_i$, we have
$$N^{\frac{1}{2}}|\hat{\alpha}_N-\at|=O_\p(1),$$
$$N^{\frac{1}{2}}\norm{S_{N,\hat{\alpha}_N}-\St}_2=O_\p(1).$$

\end{theorem}

\section{Simulation Experiments}\label{simulations}

In this section, we probe the behaviour of our models, and the finite sample performance of our estimation framework, via simulation. To generate the noise maps $T_{\epsilon_i}$, we use a class of random optimal maps introduced in \citet{ghodratidistribution} that are modifications of the maps used in \citet{panaretos2016amplitude}: Let $K$ be a random integer with a symmetric distribution around zero. We define $\zeta_K:[0,1] \to [0,1]$ by
\begin{equation}{\label{map}}
\zeta_0(x)=x, \quad \zeta_K(x)=x-\frac{\sin(\pi Kx)}{|K|\pi}, \qquad K\in Z\setminus \{0\}.
\end{equation}
These are strictly increasing smooth functions satisfying $\zeta_K(0)=0$ and $\zeta_K(1)=1$. For $x\in[0,1]$ we have $\E[\zeta_K(x)]=x$, as required in the definition of model \eqref{Model}. The random maps $T_\epsilon$ will be a mixture of the maps \eqref{map} as defined in \citet{ghodratidistribution}.

Each plot in Figure \ref{fig:simulation} corresponds to a time series simulation with a different combination of $\St$ and $\at$. Each column corresponds to a different value of $\at\in\{-0.9,-0.5,0,0.5,0.9\}$ from left to right. In the three top rows, $\St$ is chosen to be $\zeta_K$ for $K=\{-6, -4, -2\}$ from top to bottom. In row four, $\St$ is the average of $\zeta_1$ and an instance of $T_{\epsilon}$. Rows five and six exemplify the method on non-differentiable and discontinuous maps $\St$ respectively.

Plots that fall within the bounding red rectangle correspond to settings where our theory is guaranteed to apply. Plots outside of that rectangle are not guaranteed to be covered by our theory: they either distinctly violate our assumptions (such as the last row where the true map is not continuous, as required) or we cannot confirm whether the assumption \ref{stationary_solution} holds true. Starting from the identity map, we generate a time series with 300 iterations and discard the first 100 maps of the series. The remaining 200 maps $\{T_i\}$ are shown in light blue, the true map $\St$ is in dark blue, and the estimated map is in orange. For each time series, we show the estimated $\hat{\alpha}$ and the error between the estimator and true map in $\norm{.}_2$-norm.

As expected from Remark \ref{sufficient_condition_existence}, smaller values of $|\at|$ lead to time series which apparently oscillate around the mean of the stationary time series, which in turn leads to the convergence of our estimator with respect to the true map. In particular, good agreement is seen between the estimator and true map for values of $|\at|$ up to $0.5$ at least, only noticeably failing for $\at=0.5$ in the discontinuous map case (where our theoretical guarantee does not apply due to the discontinuity).

Larger values of $|\at|$ can still lead to similar stationary state time series {(sometimes even outside of the red rectangle, where our theoretical guarantees apply)} but with naturally larger oscillations. Still, a good agreement between the estimator and ground truth is observed. This can depend on the choice of map $\St$ and the precise value of $\at$. For instance, in the third, fourth, and fifth rows, when $\at=-0.9$. In the remaining rows of the first column, the stationary state behavior changes to a period-two time series (with noise) where the maps oscillate alternatively between two maps related by inversion (recall that negative values of $\at$ imply an inversion of the map $T_{i-1}$ at each time step). Nevertheless, the estimator is able to capture features of the $\St$ map that are not visible in the time series itself: notably, the discontinuous step in row six is present in the estimated map.

In the other extreme of $\at = 0.9$, the time series maps are close to step-like functions with some variation in the step height. The maps are in fact still oscillating around the mean of the stationary time series that is very close to the step-like map $\St^{\infty}$, which is the mean of the solution to the model \eqref{Model} when $\alpha \rightarrow 1$, that is $T_{s} = \St \circ T_{s}$. However, the performance of the estimator is the worst in this limit.

Do note that the family of maps $\zeta_K(x)$ is not symmetric with respect to inversion in the sense that the derivative of $\zeta_K(x)$ is $0$ at some fixed points ($\zeta_K(x) = x$) but is never infinite, and therefore the random maps $T_{\epsilon}$, which are derived from $\zeta_K(x)$, are biased in this way. For this reason, the vertical variance observed in most maps is much more pronounced than the horizontal one, which is very clear in the case $\at = 0.9$.

\begin{figure}[H]
    \centering
    \begin{tikzpicture}
    \node (f) [anchor=south west, inner sep=0pt] {\includegraphics[width=0.954\textwidth]{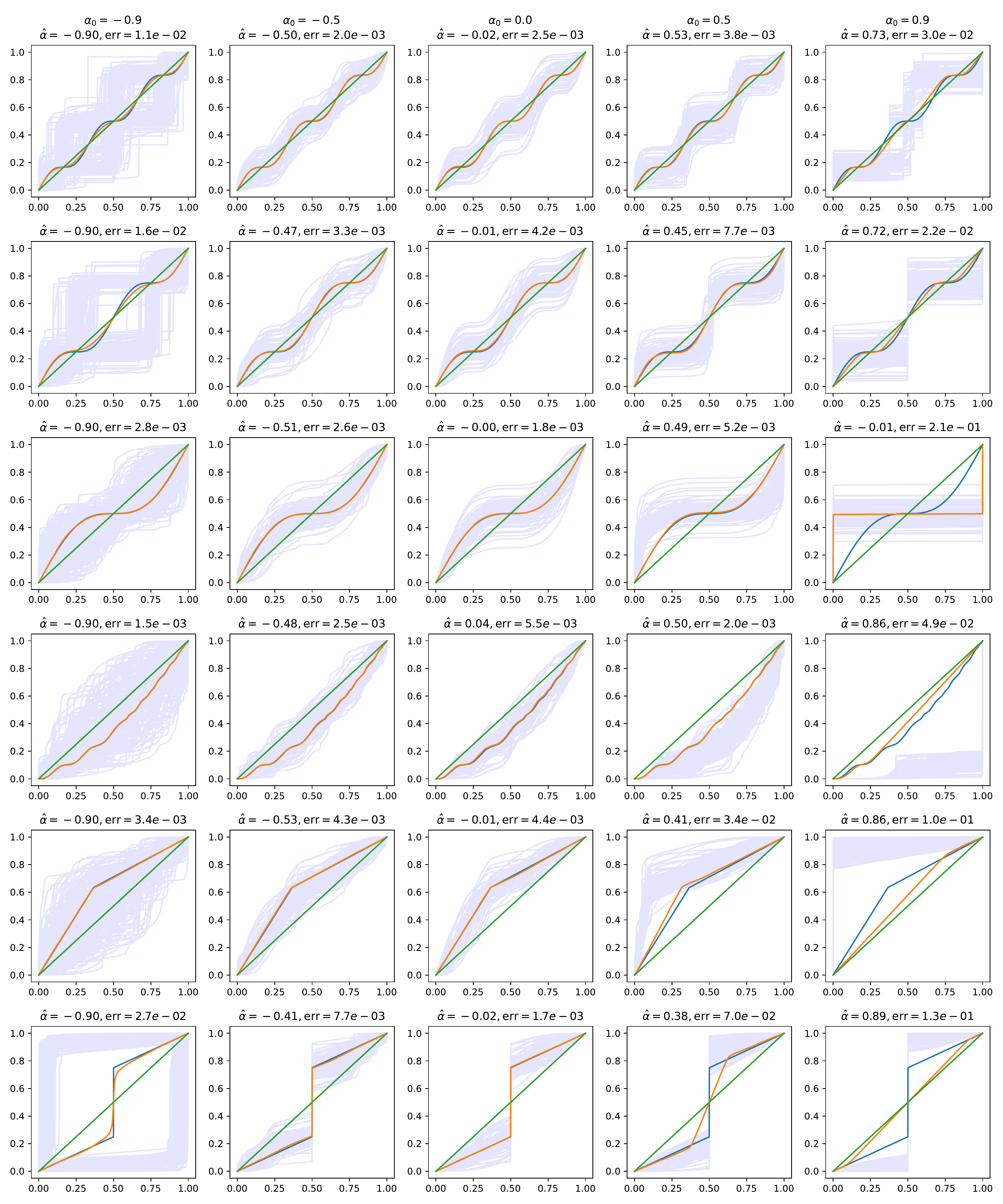}};
    \begin{scope}[x={(f.south east)},y={(f.north west)}]
    \draw[color=Red,ultra thick] (0.203, 0.16666) rectangle (0.795, 0.992);
    \end{scope}
    \end{tikzpicture}
    \caption{\small{Estimated map (orange) versus the true map (blue) for different combinations of $\at$ and $\St$. The light blue line represents the simulated time series, while the green line represents the $\id$ map. Each column corresponds to a different value of $\at$, ranging from $-0.9$ on the left to $0.9$ on the right. The top three rows show results for $\St=\zeta_K$ where $K$ is chosen from $\{-6, -4, -2\}$ from top to bottom. In the fourth row, $\St$ is the average of $\zeta_1$ and an instance of $T_{\epsilon}$. The fifth and sixth rows demonstrate the method on non-differentiable and discontinuous maps, respectively. The cases within the red rectangle are covered by our theoretical guarantees.}}
    \label{fig:simulation}
\end{figure}

\section{Illustrative Data Analysis}{\label{real_data_analysis}}
In this section, we consider the distribution of minimum daily temperatures recorded in the summer of the years from 1960 to 2020 from several airports in the USA (available at \texttt{www.ncei.noaa.gov}).  That is, the years are taken as the time index, and for any given time index we observe a distribution over the temperature scale (representing the distribution of minimal temperatures over that year's summer). Thus, each airport gives rise to a distributional time series. This data set has been also analysed by \citet{zhu2021autoregressive} to demonstrate their own distributional autoregressive model, which allows for constructive comparison.

We examine the daily minimum temperature for June, July, August, and September from 1960 to 2020 in four locations: Chicago O'Hare International Airport, Atlanta Hartsfield-Jackson International Airport, Phoenix Airport, and New Orleans Airport. The corresponding distributions are displayed in Figure \ref{fig:temperature_dist}.

\begin{figure}[H]
    \centering
    \includegraphics[width=1\textwidth]{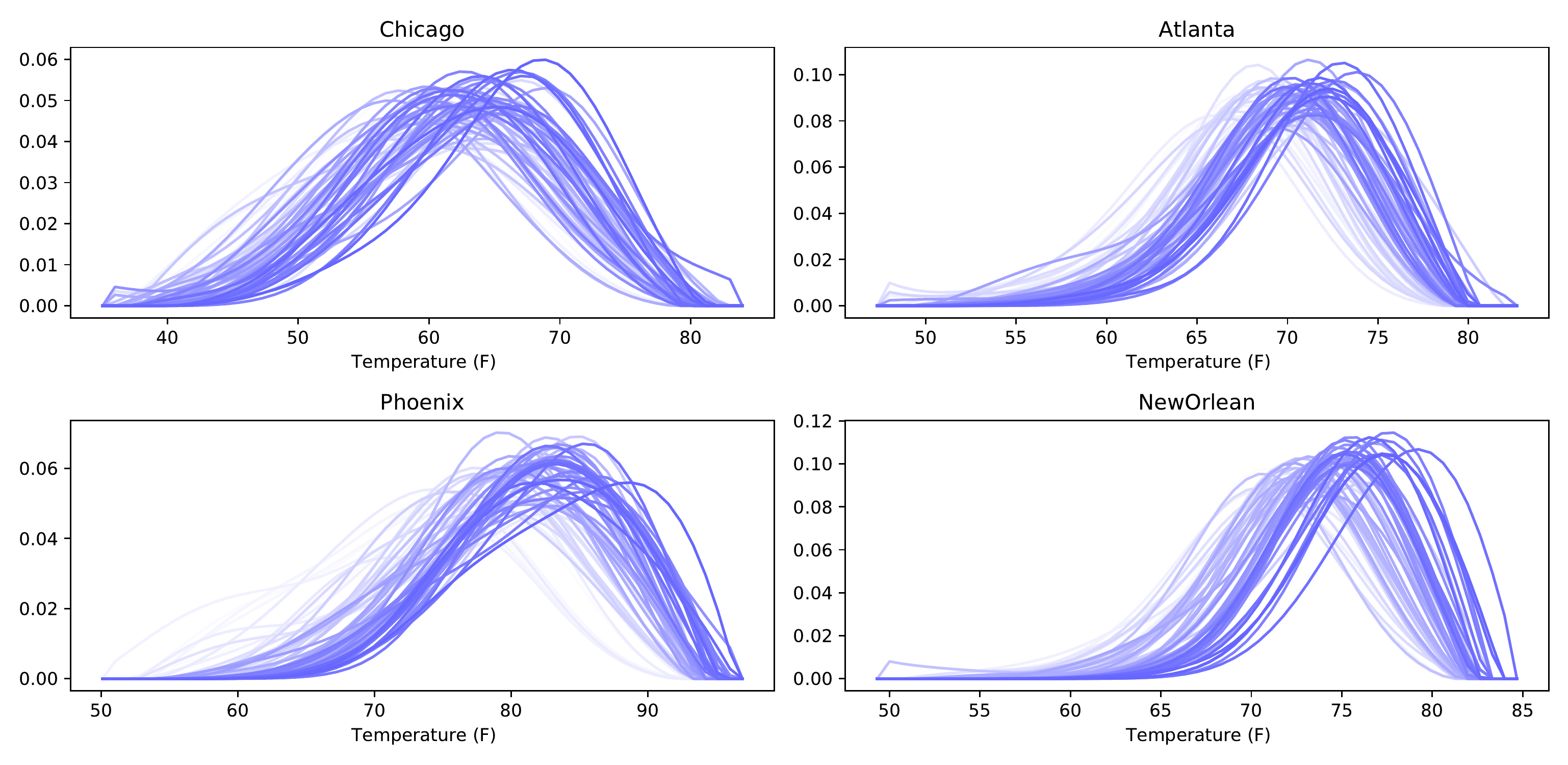}
    \caption{\small{Time series of distribution of daily minimum temperature in summer from 1960 to 2020 at Chicago Ohare international airport, Atlanta Hartsfield Jackson international airport, Phoenix airport, and New Orleans airport.} The shading reflects the time index: the fainter the curve, the earlier in time it corresponds to.}
    \label{fig:temperature_dist}
\end{figure}

The map sequence elicited by adopting the increment model (Model (I)) is shown in Figure \ref{fig:method1_Ts}. These maps are obtained by calculating the optimal maps between consecutive annual temperature distributions for each location. These maps exhibit oscillations around the identity, except in the subdomains corresponding to extreme temperature values. In the lower extreme, the maps impose a cutoff on the lower end of the support of the temperature distribution, while the higher end is pushed towards higher values and eventually reaches the extreme of the support. This implies that extreme temperatures are increasing, indicating that the coldest and hottest nights in summer are becoming hotter.

Figure \ref{fig:method1_estimates} presents the estimates of $S$ obtained using Model (I), where the estimated $\alpha$ was found to be $0$ up to three decimal points for all airports. This suggests that the optimal maps $T_i$ are independent from each other and, on average, they are equal to the estimated maps $S$ presented. The estimated $S$ maps are very similar across all airports, effectively being the identity map in the middle portion of the support and above the identity at the extreme points.

\begin{figure}[H]
    \begin{subfigure}{0.5\linewidth}
      \centering
    \includegraphics[width=\linewidth]{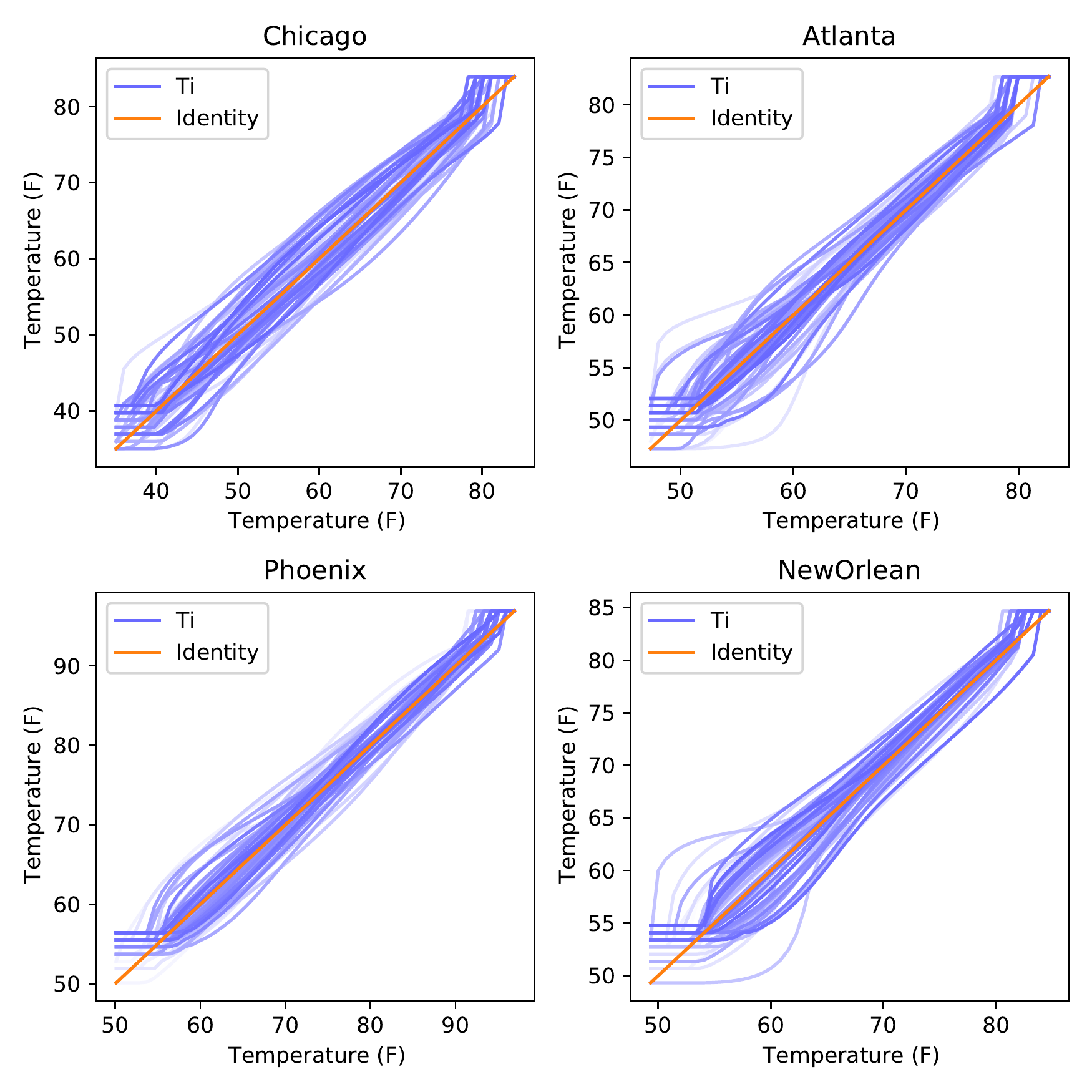}
    \caption{\footnotesize{}}
    \label{fig:method1_Ts}
    \end{subfigure}%
    \begin{subfigure}{0.5\linewidth}
      \centering
    \includegraphics[width=\linewidth]{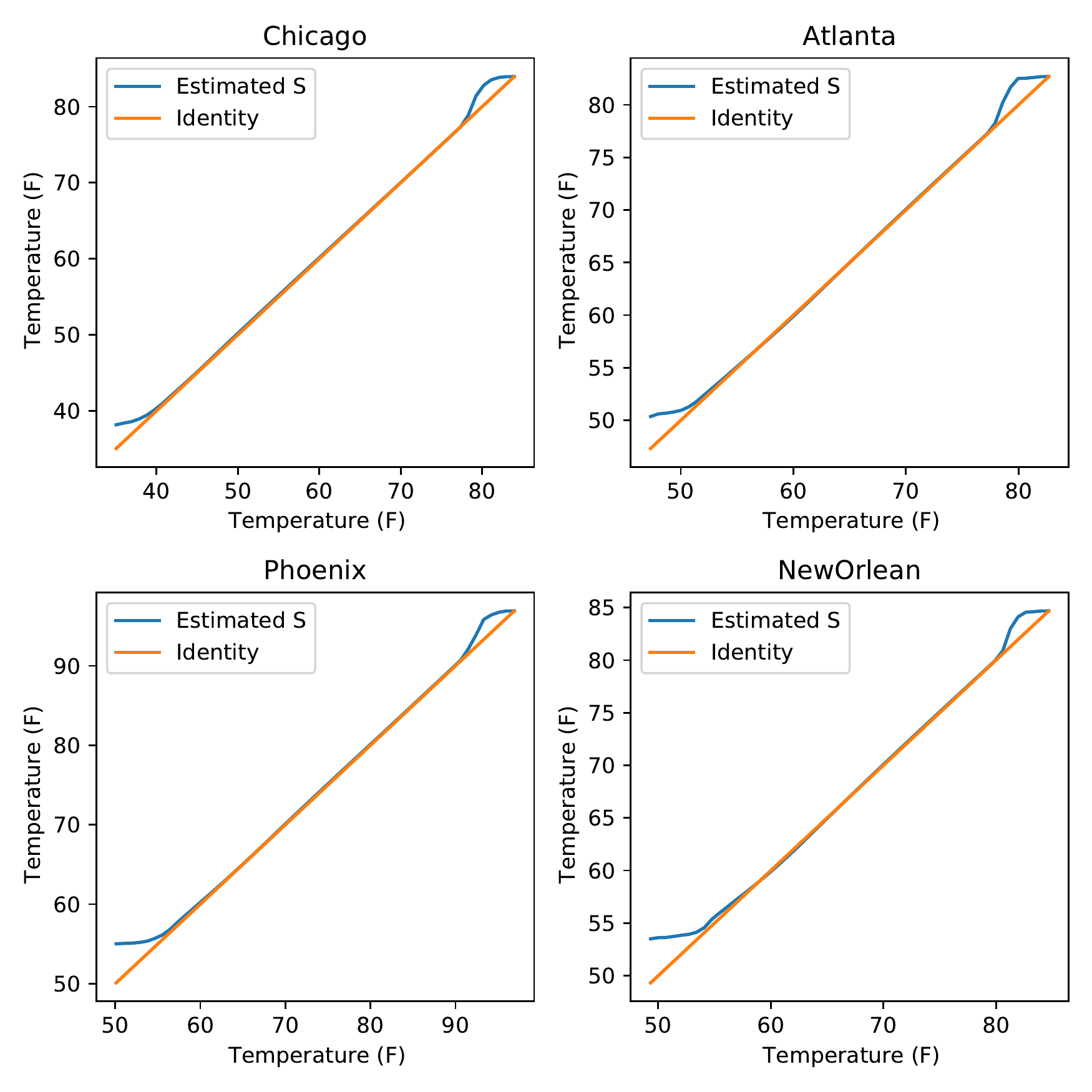}
    \caption{}
    \label{fig:method1_estimates}
    \end{subfigure}
    \caption{\small{(a): Time series $\{T_i\}$ (blue) based on model (I) and identity map (orange) for the four locations. (b): Estimated map $S$ (blue) and identity map (orange). Faint blue shading corresponds to early years, and bold shading to later years.}}
\end{figure}

Examining the maps generated by fitting Model (I), i.e. computing the optimal maps between consecutive annual distributions in Figure \ref{fig:method1_Ts}, we can observe an increasing trend in the cutoff value of the lower endpoint over time. This implies that the time series of optimal maps may not be stationary. Of course, the $S$ maps are not able to capture the overall increase in the cutoff value of the lower end over time: the plateaus of the $S$ maps are just the averages of the optimal maps $T_i$ and don't show this trend. Indeed, a problem of modeling such data is that the system may be dynamically evolving due to factors like global warming, and it is not obvious a priori if stationary regimes exist that can be captured by our models.

However, using the uniform quantile model (Model (UQ)), the resulting maps are more interpretable and reveal more refined dynamics beyond the cutoffs at the extremes. To obtain these maps, we fitted iteration \eqref{Model} to the time series of quantile functions of the temperature distributions. The quantile functions are shown in figure \ref{fig:method2_Qs}. The resulting estimated maps $S$ are in figure \ref{fig:method2_estimates}, and the estimated $\hat{\alpha}$ for the four airports are $\{0.39,0.80,0.89, 0.89\}$. All the maps show a cutoff at the lower end and a fixed point in the second half of the support where the derivative is smaller than 1. The fixed point implies a point of stability, and the derivative means there is a trend towards a concentration of weight around this point, that is, if we start the time series at a Gaussian-like distribution of mean different from the $S$ fixed point, the distributions in the time series will progress towards Gaussian-like distributions of mean approaching the fixed point. Again, the model may be failing to capture a trend of ever-increasing temperature, or it may be implying a stabilization at temperatures given by the fixed points, which will become the new norm.

\begin{figure}[H]
    \begin{subfigure}{0.5\linewidth}
      \centering
    \includegraphics[width=\linewidth]{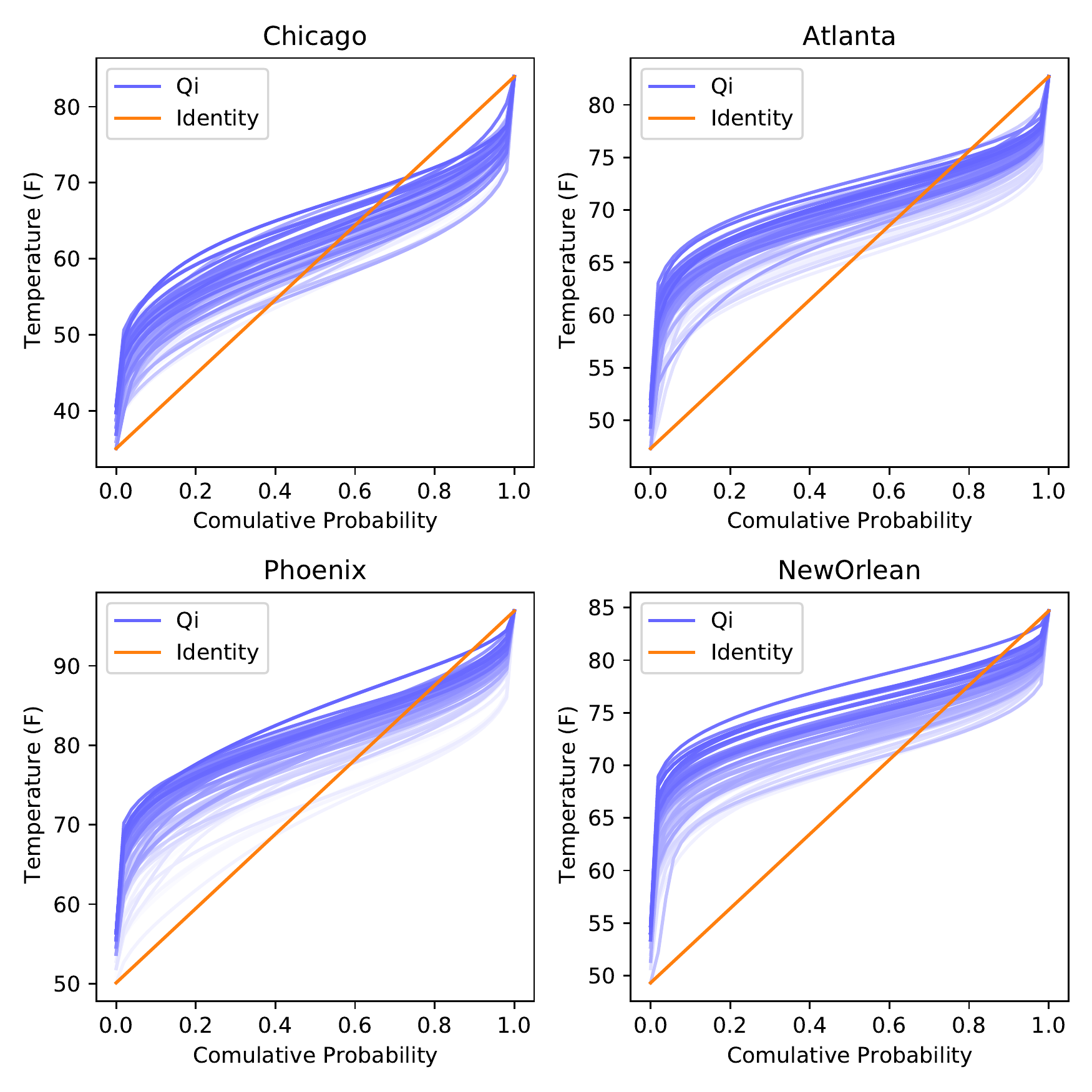}
    \caption{}
    \label{fig:method2_Qs}
    \end{subfigure}%
    \begin{subfigure}{0.5\linewidth}
      \centering
    \includegraphics[width=\linewidth]{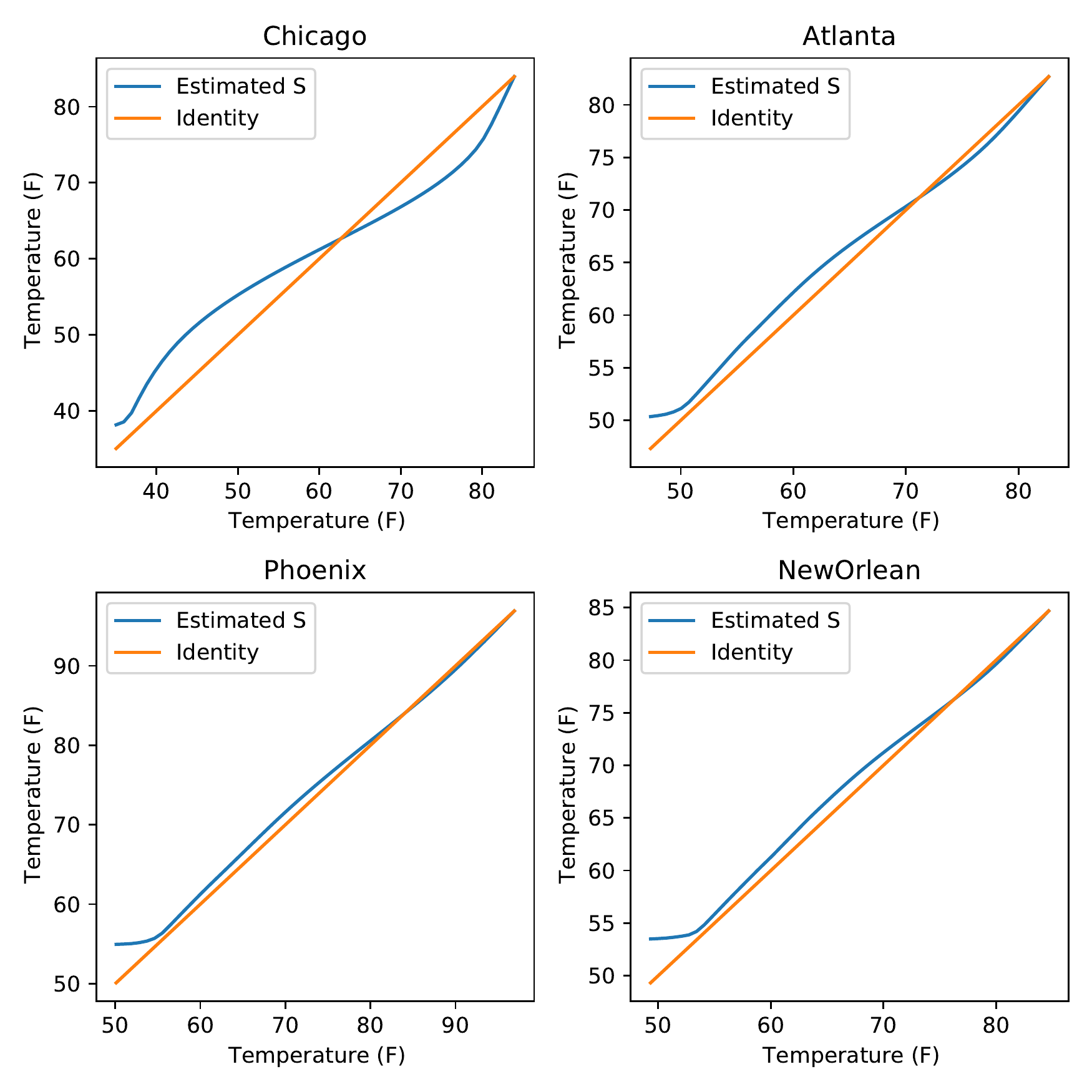}
    \caption{}
    \label{fig:method2_estimates}
    \end{subfigure}
    \caption{\small{(a): Time series of quantile functions (blue) based on model (UQ) and identity map (orange) for the four locations. (b): Estimated map $S$ (blue) and identity map (orange). Faint blue shading corresponds to early years, and bold shading to later years.}}
\end{figure}

Even if the model is possibly misspecified, the estimated maps $S$ are still able to condense several features of the time series of distributions. Namely, the reduction of extreme cold events and the progression toward higher modal temperatures which may or may not be static.

There is an interesting observation to be made given that the estimated $\alpha$ when fitting the intercept model (I) is numerically 0 while it is in (0,1) when fitting the quantile model (UQ). Specifically, in combination, these results suggest that the quantile model is, in a certain sense, a better fit to the data. The reasoning is as follows. Recall that the increment model (I) with $\alpha = 0$ is equivalent to the quantile model (UQ) when $\alpha = 1$, and corresponds to ``trivial dynamics" (random walk). With those respective values of $\alpha$, the two models yield the \textit{same} estimator for $S$, namely the map $G=\frac{1}{N}\sum_{j=1}^{N}F_{i}^{-1}\circ F_{i-1}$ (see equation \eqref{estimator-components}, where $T_i=F_{i}^{-1}\circ F_{i-1}$ for model (I), whereas $T_i=F_{i}^{-1}$ for model (UQ)). Since the estimated $\alpha$ is zero under the increment model (I), then the best fitting model of type (I) yields a fit
$$M_N^{(I)}(0)=\sum_{i} \| G - F_i^{-1}\circ F_{i-1} \|^2_2=\sum_{i} \| G\circ F^{-1}_{i-1} - F_i^{-1} \|^2_2=M_N^{(UQ)}(1).$$
The last expression on the right-hand side is interpretable as the fit obtained under the (UQ) model when estimating $\alpha$ by 1. But this is strictly worse than the best fit, which is obtained at values of $\alpha$ distinctly smaller than 1, leading to non-trivial dynamics (as opposed to those corresponding to a random walk). In other words, the best possible fit under the increment model can be interpreted in the same sense as the best fit in the quantile model and is strictly worse in that sense.

A more high-level way of seeing this is to say that whenever fitting model (I) results in an estimated $\alpha$ that is nearly zero, then the best fitting model of type (I) is in fact a (UQ) model. In which case we have evidence to prefer a (UQ) modeling approach instead, which will correspond to non-trivial dynamics. Conversely, if fitting model (UQ) yields an estimated $\alpha$ near 1, it may be preferable to use model (I) instead.

\section{Proofs}

\begin{proof}[Proof of Lemma \ref{stationary_solution}]
    The proof is directly analogous to that of Theorem 2 in \citet{wu2004limit} and theorem 1 in \citet{zhu2021autoregressive}.
\end{proof}

\begin{proof}[Proof of Lemma \ref{minimizer_expectation}]
    We prove the theorem in the following 4 steps:

    \begin{enumerate}
        \item Given a function $f\in L^2$, and a random function $\epsilon$ such that $\E(\epsilon)= \id$, we can show that
        $$\arg\min_h \E_\epsilon\norm{h-\epsilon\circ f}^2_2=f.$$
        To do so, we can apply Fubini's theorem and rewrite the expression as follows:
        $$\int\int |h(x)-\epsilon(f(x))|^2 \diff x \diff \epsilon= \int\int|h(x)-\epsilon(f(x))|^2 \diff \epsilon \diff x $$
        Since $\E_\epsilon [\epsilon(f(x))] = f(x)$ for any $x$, the minimizer of the inner integral on the left-hand side is $h(x) = \E[\epsilon(f(x))]=f(x)$.

        \item We will now demonstrate that for any fixed $T_i$ and $T_{i-1}$, as well as for all $\alpha$, the following inequality holds:
        $$\E_\epsilon [g_{\at}(T_{i-1},T_i,S_{\at})] \leq \E_\epsilon [g_\alpha(T_{i-1},T_i,S_\alpha)].$$

        Let us define $f(\alpha, T) = S_\alpha \circ [\alpha T]$. Note that for all indices $i$, we have
        \begin{equation}\label{g_form}
        \begin{split}
          g_\alpha(T_i,T_{i-1},S_\alpha)&=\norm{S_\alpha \circ [\alpha T_{i-1}]-T_{\epsilon_i}\circ S_{\at} \circ [\at T_{i-1}]}^2_2\\
            &=\norm{f(\alpha,T_{i-1})-T_{\epsilon_i}\circ f(\at,T_{i-1})}^2_2.
        \end{split}
        \end{equation}

    Using the result from part 1 and the equation (\ref{g_form}), we can conclude that if there exists an $\alpha$ such that $f(\alpha,T_{i-1})$ minimizes the expression $\E_\epsilon [g_{\at}(T_{i-1},T_i,S_{\at})]$ in equation \ref{g_form}, then we must have $f(\alpha,T_{i-1}) = f(\at,T_{i-1})$. Therefore, we obtain the desired inequality.

\item We now aim to prove that for any $\alpha$, we have
$$\E[g_{\at}(T_{i-1},T_i,S_{\at})] \leq \E[g_\alpha(T_{i-1},T_i,S_\alpha)].$$

We start by denoting by $\pi$ the marginal distribution of $T_i$, and $Q$ the marginal distribution of the pair $(T_{i-1},T_{i})$. Then, we can express the expectation of $g_\alpha(T_{i-1},T_i,S_\alpha)$ as follows:
$$\E_Q [g_\alpha(T_{i-1},T_i,S_\alpha)] = \E_\pi [\E_\epsilon{g_\alpha(T_{i-1},T_i,S_\alpha)|T_{i-1}}].$$

By using part 2 of the proof, we know that $\at$ is a minimizer for the inner expectation of the right-hand side, i.e.,
$$\E_\epsilon\{g_{\at}(T_{i-1},T_i,S_{\at})|T_{i-1}\} \leq \E_\epsilon\{g_\alpha(T_{i-1},T_i,S_\alpha)|T_{i-1}\},$$
and this for all $T_{i-1}$. Therefore, taking the expectation over $T_{i-1}$, we get
$$\E_Q [g_{\at}(T_{i-1},T_i,S_{\at})] \leq \E_Q [g_\alpha(T_{i-1},T_i,S_\alpha)].$$

        \item Finally we can conclude that $\at$ is the unique minimizer of $M(\alpha)$. Suppose there exists an $\alpha$ such that $\E[g_{\at}(T_{i-1},T_i,S_{\at})] = \E[g_\alpha(T_{i-1},T_i,S_\alpha)]$. Using parts 2 and 3, we can deduce that for each fixed $T_i, T_{i-1}$, $\E_\epsilon [g_{\at}(T_{i-1},T_i,S_{\at})] = \E_\epsilon [g_\alpha(T_{i-1},T_i,S_\alpha)]$. Then using equation \ref{g_form} we can conclude that, for all indices $j$,
       $$  S_\alpha \circ [\alpha T_j] = S_{\at} \circ[\at T_{j}].$$
        If $S_\alpha \circ \alpha T_j = S_{\at} \circ[\at T_{j}]$ for all $j$, we can deduce $S_\alpha = S_{\at} \circ[\at T_{j}] \circ [\alpha T_j]^{-1}$ for all $j$. However, note that while $S_\alpha$ is deterministic, the right-hand side is deterministic (and not random) if and only if $\alpha = \at$. This is because if $\alpha \neq \at$, then the right-hand side depends on $T_j$, which is a random variable.

    \end{enumerate}
\end{proof}

\begin{lemma}{\label{norm_inverse}}
For any $T,S\in \mathcal{T}$ we have $\norm{T^{-1}-S^{-1}}_2\lesssim \sqrt{\norm{T-S}_2}$.
Moreover, let $\mathcal{T}_{l,u} = \{ T\in \mathcal{T}:  0<L_l\leq T' \leq L_u <\infty\}$. For any $T,S \in \mathcal{T}_{l,u}\subset \mathcal{T}$ we have $\norm{T^{-1}-S^{-1}}_2 \lesssim \norm{T-S}_2$. In summary,  there exists $b\in[\frac{1}{2},1]$ such that  $\norm{T^{-1}-S^{-1}}_2\lesssim \norm{T-S}_2^b$ for any $T,S \in \mathcal{T}$.
\end{lemma}

\begin{proof}

Let $T,S \in\mathcal{T}$. Then, for some constant $C'$, we have: $\norm{T^{-1}-S^{-1}}_2\leq C' \norm{T^{-1}-S^{-1}}_1$, because the functions are bounded. Moreover,
$\norm{T^{-1}-S^{-1}}_1=\norm{T-S}_1$. And, finally, by applying the Cauchy-Schwarz inequality, we get
$\norm{T-S}_1\leq C\sqrt{\norm{T-S}_2}$, where $C$ is a constant. Therefore, we conclude $\norm{T^{-1}-S^{-1}}_2\leq C\sqrt{\norm{T-S}_2}$.

When $T,S\in \mathcal{T}_{l,u}$ we can write
\begin{equation}
\begin{split}
\norm{T^{-1}-S^{-1}}^2_2 &= \int_0^1 |T^{-1}(x)-S^{-1}(x)|^2\diff x\\
& = \int_0^1 |T^{-1}\circ S(y)-y|^2 S'(y)\diff y \quad\quad (S^{-1}(x)=y)\\      
& \leq L_u \int |T^{-1}\circ S(y) - y|^2 \diff y \\
& \leq L_u \int |z-S^{-1}\circ T(z)|^2 \frac{1}{S'(S^{-1}\circ T(z))} T'(z)\diff z \quad\quad (T^{-1}\circ S(y)=z)\\
&\leq L_u \frac{L_u}{L_l} \int |z-S^{-1}\circ T(z)|^2 \diff z \\
&\leq L_u \frac{L_u}{L_l} \frac{1}{L_l} \int |S(z)-T(z)|^2 \diff z \quad\quad (\forall x,y \quad |x-y|\leq \frac{1}{L_l}|S(x)-S(y)|)\\            
&\leq \frac{L_u^2}{L_l^2} \norm{S-T}^2_2.
\end{split}
\end{equation}

\end{proof}

\begin{lemma}{\label{S_g_Lipschitz}}
There exists a constant $\frac{1}{2}\leq b \leq 1$ such that the following inequalities hold:
$$ \norm{S_{\alpha_1}-S_{\alpha_2}}_2\lesssim|\alpha_1-\alpha_2|^b,$$
and
$$ \norm{S_{N,\alpha_1}-S_{N,\alpha_2}}_2 \lesssim |\alpha_1-\alpha_2|^b,$$
and
$$g_{\alpha_1}(T_{i-1},T_i,S_{N,\alpha_1})-g_{\alpha_2}(T_{i-1},T_i,S_{N,\alpha_2})\leq C(T_i)|\alpha_1-\alpha_2|^b,$$
where $\frac{1}{n}\sum_i \E [ C(T_i)] = O(1)$.

Define $\mathcal{T}_{l,u} = \{ T\in \mathcal{T}:  0<L_l\leq T' \leq L_u <\infty\}$. If $\{T_i\} \subset \mathcal{T}_{l,u}$, then $b=1$ in the above inequalities.
\end{lemma}
\begin{proof}
To begin with, it should be noted that given any two real numbers $\alpha_1,\alpha_2\in(-1,1)$ with the same sign, and for any given map $T$, we have the following inequality:
$$\norm{[\alpha_1 T] - [\alpha_2 T]}_2\leq C|\alpha_1-\alpha_2|,$$
where $C$ is a constant. In fact, it suffices to consider the definition of $[\alpha T]$ for the cases when $\alpha\geq 0$ and $\alpha<0$ separately.
Using Lemma \ref{norm_inverse} we can write that for some $\frac{1}{2}\leq b \leq 1$,
\begin{equation}
\begin{split}
\norm{S_{\alpha_1}-S_{\alpha_2}}_2&=\norm{\E[T_j \circ [\alpha_1 T_{j-1}]^{-1}]-\E[T_j \circ [\alpha_2 T_{j-1}]^{-1}]}_2\\
&\leq L C'|\alpha_1-\alpha_2|^b,
\end{split}
\end{equation}
where $L$ is the common Lipschitz constant for all $T_j$. Similarly
$$ \norm{S_{N,\alpha_1}-S_{N,\alpha_2}}_2 \leq L |\alpha_1-\alpha_2|^b,$$

We now proceed to show that both $S_{N,\alpha}$ and $S_\alpha$ are Lipschitz functions of $x$. To do this, we observe that the inverse of a Lipschitz function is Lipschitz, and also the composition of two Lipschitz functions is Lipschitz. Since all $T_j$ are Lipschitz and $S_{N,\alpha}$ and $S_\alpha$ are defined as compositions, they are also Lipschitz with respect to $x$.

\vspace{5mm}

We will now show that $g$ is Lipschitz function of $\alpha$:
\begin{equation}
\begin{split}
g_{\alpha_1}(T_{i-1},T_i,S_{N,\alpha_1})-g_{\alpha_2}(T_{i-1},T_i,S_{N,\alpha_2}) &= \norm{S_{N,\alpha_1} \circ [\alpha_1 T_{i}] - T_{i+1}}^2-\norm{S_{N,\alpha_2} \circ [\alpha_2 T_{i}] - T_{i+1}}^2_2\\
&\lesssim \norm{S_{N,\alpha_1} \circ [\alpha_1 T_{i}] - S_{N,\alpha_2} \circ [\alpha_2 T_{i}]}_2\\
&\lesssim \norm{S_{N,\alpha_1}\circ [\alpha_1 T_{i}] - S_{N,\alpha_1} \circ [\alpha_2 T_{i}]}_2\\
&\quad +\norm{S_{N,\alpha_1} \circ [\alpha_2 T_{i}] - S_{N,\alpha_2} \circ [\alpha_2 T_{i}]}_2\\
&\leq D(T_i)|\alpha_1-\alpha_2| + C(T_i)|\alpha_1-\alpha_2|^b\\
&\lesssim C(T_i)|\alpha_1-\alpha_2|^b
\end{split}
\end{equation}
where $D(T_i)$ and $C(T_i)$ are constants that depend on $T_i$ and $\sum_i \E [ C(T_i)] /n = O(1)$.
\end{proof}

\begin{proof}[Proof of Lemma \ref{m_dependent}]
    Let $T_{n-m} = f(\epsilon_{n-m},\epsilon_{n-m-1},\cdots,)$ and $T'_{n-m} = f(\epsilon'_{n-m},\epsilon'_{n-m-1},\cdots,)$. Thus we can write $T_n = \Phi_{n,m}(T_{n-m})$ and $T'_{n,m} = \Phi_{n,m}(T'_{n-m})$.

    \begin{equation}
        \begin{split}
            \sum_{m=1}^\infty (\E\norm{T_n-T'_{n,m}}_2^{2})^{1/2} &= \sum_{m=1}^\infty (\E \norm{\Phi_{n,m}(T_{n-m}) - \Phi_{n,m}(T'_{n-m})}_2^2)^{1/2}\\
            & \leq \sum_{m=1}^\infty (\E\norm{\Phi_{n,m}(T_{n-m})-\Phi_{n,m}(Q_0)}_2^2)^{1/2} \\
            &\quad \quad+ (\E\norm{\Phi_{n,m}(T'_{n-m})-\Phi_{n,m}(Q_0)}_2^2)^{1/2}\\
            (\text{Lyapunov's inequality}) \quad   & \leq \sum_{m=1}^\infty (\E\norm{\Phi_{n,m}(T_{n-m})-\Phi_{n,m}(Q_0)}_2^\eta)^{1/\eta} \\
            &\quad \quad+ (\E\norm{\Phi_{n,m}(T'_{n-m})-\Phi_{n,m}(Q_0)}_2^\eta)^{1/\eta}\\
 \text{(Assumption \ref{moment_contracting_assumption})} \quad\quad\quad \quad\quad         &\leq \sum_{m=1}^\infty C r^{m/\eta} (\norm{T_{n-m}-Q_0}_2\vee \norm{T'_{n-m}-Q_0}_2)\\
            &<\infty
        \end{split}
    \end{equation}
\end{proof}

The following statement is virtually obvious, but is used multiple times in the proofs below and so is most easily quoted directly:

\begin{lemma}{\label{CLTtoLLN}}
Let $\{X_i\}$ be a sequence of random variables, and suppose that
$W_{n}:= \sqrt{n} (\frac{1}{n}\sum_{i=1}^{n} X_{i} - \mu) \overset{d}{\to} W$ for some (almost surely finite) random variable $W$.
Then, $\frac{1}{n}\sum_{i=1}^{n} X_{i}$ converges in probability to $\mu$.

\end{lemma}
\begin{proof}
By Slutsky's Theorem, we get $n^{-1/2} W_{n} \overset{d}{\to} 0$, which also implies convergence in probability to zero.
\end{proof}

\begin{proof}[Proof of Lemma \ref{CLT_maps}]
    We start by using Lemma \ref{m_dependent} to conclude that the series $\{T_i-\E T_i\}_{i=-\infty}^\infty$ satisfies the assumptions (1.1),(1.2),(2.1) and (2.2) of \citet{horvath2013estimation}. From this, we can argue that the series $\{T_j\circ [\alpha T_{j-1}]^{-1} - \E T_j\circ [\alpha T_{j-1}]^{-1} \}_{i=-\infty}^\infty$ also satisfies those assumptions and therefore we obtain the following central limit theorem for $S_{N,\alpha}$: for any $\alpha$, there is a Gaussian process $\Gamma_\alpha$ such that
    $$\sqrt{N} (S_{N,\alpha}-S_\alpha)  \overset{d}{\to} \Gamma_\alpha,  \quad \text{in } L^2.$$
    Using the central limit theorem and Lemma \ref{CLTtoLLN}, we can infer the convergence in probability of $S_{N,\alpha}$ to $S_\alpha$ for any $\alpha$ (in $L^2$). Since both $S_\alpha$ and $S_{N,\alpha}$ are globally Lipschitz with respect to $\alpha$, in the sense of Lemma \ref{S_g_Lipschitz}, we can use Corollary 3.1 of \citet{newey1991uniform} to obtain uniform convergence in probability:
    $$\sup_\alpha \norm{S_{N,\alpha}-S_\alpha}_2 \to 0 \quad \text{in probability}.$$
\end{proof}

\subsection{Overview of \citet{wu2004limit}}
In their work, \citet{wu2004limit} investigated the properties of nonlinear time series expressed in terms of iterated random functions and established a central limit theorem for additive functionals of such systems. The construction involves a sequence of functions of the form $X_n(x)=F_{\theta_n}\circ F_{\theta_{n-1}}\circ \cdots F_{\theta_1}(x)$. The authors assume that $X_n$ satisfies a geometric moment condition, which requires the existence of $\beta>0$, $C=C(\alpha)>0$, and $r=r(\alpha)\in (0,1)$ such that, for all $n\in N$,
\begin{equation}\label{ineq2wu}
    \E\{\rho(X_n(X'_0),X_n(X_0))^\beta\}\leq Cr^n.
\end{equation}

In addition, they define the $l$-dimensional vector $Y_i=(X_{i-l+1},X_{i-l+2},\cdots,X_i)$ and for any $\delta>0$, they introduce the functional $\Delta_g(\delta)$ as
$$\Delta_g(\delta)=\sup\{\norm{[g(Y)-g(Y_1)]1_{\rho(Y,Y_1)}\leq \delta}: \quad Y,Y_1 \quad \text{are identically distributed}\},$$
Where $\rho(.,.)$ is the product metric and is defined as $$\rho(z,z')=\sqrt{\sum_{i=1}^l \rho(z_i,z_i')^2} \quad \text{ for } z=(z_1,\cdots,z_l),z'=(z_1',\cdots,z_l').$$
Finally, the functional $S_{n,l}(g) = \sum_{i=1}^n g(X_{i-l+1},X_{i-l+2},\cdots,X_i)$ is defined. The authors establish the following central limit theorem for this functional:

\begin{theorem}\label{theorem3}
    (\citet[Theorem 3]{wu2004limit})
    Assume that (\ref{ineq2wu}) holds, that $X_1\sim\pi$, $E\{g(Y_1)\}=0$, and $E\{|g(Y_1)|^p\}<\infty$ for some $p>2$, and that
    \begin{equation}\label{stoch_dini_cont}
        \int_0^1 \frac{\Delta_g(t)}{t}<\infty.
    \end{equation}
    Then there exists a $\sigma_g\geq 0$ such that, for $\pi$-almost $x$, $\{S_{\floor{nu},l}(g)/\sqrt{n}, 0\leq u\leq 1\}$ conditional on $X_0=x$, converges to $\sigma_g B$, where $B$ is a standard Brownian motion.
\end{theorem}

A function that satisfies \eqref{stoch_dini_cont} is referred to as \textit{stochastic Dini continuous}. Using Theorem \ref{theorem3} to derive a central limit theorem for $M_N$ poses a problem: Theorem \ref{theorem3} uses fixed-length sub-sequences of the time series, i.e., $(X_{i-l+1},X_{i-l+2},\cdots,X_i)$, as arguments for the function $g$, however the arguments of the function $g$ that appears in the expression of $M_N$ in \ref{estimator}, include not only  $(T_{i-1},T_i)$, but also $S_{N,\alpha}$, thus making it dependent on the entire time series. Therefore, Theorem \ref{theorem3} cannot be applied directly, and a modified version is required. We present a modified version of Theorem \ref{theorem3} that is specifically tailored for functions of finite dimensional random variables, followed by another modification that is suitable for functionals of infinite dimensional variables.

\begin{corollary}(Modified version of \citet[Theorem 3]{wu2004limit} for finite dimensional arguments)\label{theorem3_modified2}
Suppose $\overline{Z}_n$ is a measurable function of $(X_1,X_2,\cdots,X_n)$ such that $\overline{Z}_n$ converges in probability to some constant $\mu$. Let $Y_i=(X_{i-l+1},X_{i-l+2},\cdots,X_i)$, and assume that $g(Y_i,\mu)$ is differentiable with respect to its second argument and that both $g(Y_i,\mu)$ and the derivative of $g(Y_i,\mu)$ with respect to its second argument satisfy the conditions of Theorem \ref{theorem3}. Then there exists $\sigma_g\geq0$ such that  
$$\frac{S_n}{\sqrt{n}} \to N(0,\sigma_g^2),$$
where $S_n = \sum_{i=1}^n g(Y_i,\overline{Z}_n)$.
   \end{corollary}
\begin{proof}
    By Taylor expansion, we write
    $$g(Y_i,\overline{Z}_n) = g(Y_i,\mu)+g^{(0,1)}(Y_i,\mu) (\overline{Z}_n-\mu) + \text{higher order terms}.$$
    Since $g^{(0,1)}(Y_i,\mu)$ is only a function of $Y_i$ and of a constant $\mu$, by Theorem \ref{theorem3} we have
    $$\frac{1}{\sqrt{n}} \sum_{i=1}^n \big[g^{(0,1)}(Y_i,\mu) - \E g^{(0,1)}(Y,\mu)\big] \to N(0,\sigma_{g'}^2),$$
    where $Y \stackrel{D}{\sim} Y_i$.    This implies that if $N \stackrel{D}{\sim} N(0,\sigma_{g'}^2)$,we then have
    \begin{equation}
        \begin{split}
            \frac{1}{\sqrt{n}} S_n &= [\frac{1}{\sqrt{n}}\sum_{i=1}^n g(Y_i,\mu)] + (N+ \sqrt{n} \E g^{(0,1)}(Y,\mu))(\overline{Z}_n - \mu)\\
            &=\frac{1}{\sqrt{n}}\sum_{i=1}^n [g(Y_i,\mu)+\E g^{(0,1)}(Y,\mu)(Z_i-\mu) ]+N(\overline{Z}_n - \mu).
        \end{split}
    \end{equation}
    Since $N(\overline{Z}_n - \mu) = o_\p(1),$ applying Theorem \ref{theorem3} we get
    $$\frac{1}{\sqrt{n}}\sum_{i=1}^n [g(Y_i,\mu)+\E g^{(0,1)}(Y,\mu)(Z_i-\mu) ]\to N(0,\sigma_g^2)$$
\end{proof}

\begin{corollary}(Modified version of \citet[Theorem 3]{wu2004limit} for infinite dimensional arguments)\label{theorem3_modified}
Suppose $\overline{Z}_n$ is a measurable function of $(X_1,X_2,\cdots,X_n)$ such that $\overline{Z}_n$ converges in probability to some constant $\mu$. Let $Y_i=(X_{i-l+1},X_{i-l+2},\cdots,X_i)$, and assume that $g(Y_i,\mu)$ is Fr\'echet differentiable with respect to its second argument, and that both $g(Y,\mu)$ and the Fr\'echet derivative of $g$ with respect to its second argument satisfy the conditions of Theorem \ref{theorem3}. Then there exists $\sigma_g\geq0$ such that  
    $$\frac{S_n}{\sqrt{n}} \to N(0,\sigma_g^2),$$
    where $S_n = \sum_{i=1}^n g(Y_i,\overline{Z}_n)$.
\end{corollary}

\begin{remark}
The proof of this Corollary can be understood by following the same steps as in the proof of Corollary \ref{theorem3_modified2}, without the added technical complexities that arise when dealing with the Fr\'echet derivative.
\end{remark}

\begin{proof}[Proof of Corollary \ref{theorem3_modified}]

    Let $D_g(Y_i,u,v)$ denote the Fr\'echet derivative of $g$ with respect to its second argument at $u$ in the direction $v$. Assume $\overline{Z}_n = \mu + v_n$, and apply the Taylor formula for the Fr\'echet derivative (\citet{kurdila2006convex}) to get
    $$g(Y_i,\overline{Z}_n) = g(Y_i,\mu)+D_g(Y_i,\mu,v_n) + R(Y_i,\mu,v_n),$$
    where
    $$\lim_{\norm{v_n}\to 0} \frac{|R(Y_i,\mu,v_n)|}{\norm{v_n}}=0.$$

    Note that we can identify the Fr\'echet derivative with a bounded linear
    operator as
    $$D_g(Y_i,\mu,v_n) = \langle D_g(Y_i,\mu),v_n\rangle.$$
Furthermore, as the Fr\'echet derivative is also stochastic Dini continuous, we can apply Theorem \ref{theorem3} to obtain
    $$\frac{1}{\sqrt{n}} \sum_{i=1}^n \big[\langle D_g(Y_i,\mu),v_n \rangle - \E \langle D_g(Y,\mu),v_n \rangle \big] \to N(0,\sigma_{g'}^2),$$
    where $Y \stackrel{D}{\sim} Y_i$.

     This implies that if $N \stackrel{D}{\sim} N(0,\sigma_{g'}^2)$, using the fact that the mapping $D_g(Y_i,\mu,.)$ is linear, we get:
    \begin{equation}
        \begin{split}
            \frac{S_n}{\sqrt{n}} &= \frac{1}{\sqrt{n}}\sum_{i=1}^n [g(Y_i,\mu)+D_g(Y_i,\mu,\overline{Z}_n-\mu)]\\
            &= \frac{1}{\sqrt{n}}\sum_{i=1}^n [g(Y_i,\mu)+\langle D_g(Y_i,\mu),\overline{Z}_n-\mu \rangle ]\\
            & =\frac{1}{\sqrt{n}}\sum_{i=1}^n [g(Y_i,\mu)]+\langle N \times \id +\sqrt{n}\E D_g(Y,\mu),\overline{Z}_n-\mu \rangle \\
            & = \frac{1}{\sqrt{n}}\sum_{i=1}^n [g(Y_i,\mu)+\langle \E D_g(Y,\mu),Z_i-\mu \rangle] + N \langle \id ,\overline{Z}_n-\mu \rangle
        \end{split}
    \end{equation}

    Since $N \langle \id ,\overline{Z}_n-\mu \rangle = o_P(1),$ and $\E \langle \E D_g(Y,\mu),Z_i-\mu \rangle = 0$, we can apply Theorem \ref{theorem3} and conclude $\frac{S_n}{\sqrt{n}}  \to N(0,\sigma^2)$ for some $\sigma$.
\end{proof}

\begin{lemma}{\label{frechet}}
    The function $g_\alpha(T_{i-1},T_i,S) = \norm{S\circ [\alpha T_{i-1}]-T_i}_2^2$ is Fr\'echet differentiable with respect to $S$ and satisfies the Taylor formula
    $$g_\alpha(T_{i-1},T_i,S+v) = g_\alpha(T_{i-1},T_i,S)+D_g(T_{i-1},T_i,S,v) + R(T_{i-1},T_i,S,v),$$
    where $D_g(T_{i-1},T_i,S,v)$ is the Fr\'echet derivative of $g_\alpha$ with respect to $S$ in the direction $v$, and
    $$\lim_{\norm{v}\to 0} \frac{|R(T_{i-1},T_i,S,v)|}{\norm{v}}=0.$$
    Furthermore, the mapping $D_g(T_{i-1},T_i,\mu,.)$ is both linear and bounded.
\end{lemma}

\begin{proof}
    To begin, we show that $g_\alpha(T_{i-1},T_i,S)$ is Gateaux differentiable.
    \begin{equation}
        \begin{split}
            \lim_{\epsilon \to 0} \frac{g_\alpha(T_{i-1},T_i,S+\epsilon v)-g_\alpha(T_{i-1},T_i,S)}{\epsilon}&=
            \lim_{\epsilon \to 0} \frac{\norm{(S+\epsilon v) \circ [\alpha T_{i-1}]-T_i}_2^2-\norm{S \circ [\alpha T_{i-1}]-T_i}_2^2}{\epsilon}\\
            &=  \lim_{\epsilon \to 0} \frac{\epsilon^2 \norm{v\circ [\alpha T_{i-1}]}_2^2+\epsilon \langle v\circ \alpha T_{i-1}, S\circ [\alpha T_{i-1}]-T_i \rangle }{\epsilon}\\
            &= \langle v\circ [\alpha T_{i-1}],S\circ [\alpha T_{i-1}]-T_i \rangle\\
            &= D_g((T_i,T_{i-1}),S,v).
        \end{split}
    \end{equation}
  As the above expression is linear and bounded with respect to $v$, it serves as the Gateaux differential. As $D_g(T_{i-1},T_i,S)$ is Gateaux differentiable for every $T$ and the mapping $T\to D_g(T_{i-1},T_i,S)$ is continuous, Corollary 4.1.1. of \cite{kurdila2006convex} guarantees that $D_g$ is also the Fr\'echet derivative.
\end{proof}

\begin{lemma}{\label{SDC}}
The stochastic Dini continuity condition \eqref{stoch_dini_cont} is satisfied by the function $g_\alpha$.
\end{lemma}

\begin{proof}

We want to show $\int_0^1 \frac{\Delta_g(t)}{t}<\infty$, where
\begin{equation}
\begin{split}
\Delta_g(\delta)=\sup \Big\{\norm{[g_\alpha(T_{i-1},T_i,S)-g_\alpha(T'_{i-1},T'_i,S)]  1_{\rho((T_{i-1},T_i),(T'_{i-1},T'_i))\leq \delta}} \\
:T_i,T'_i \text{ are identically distributed}\Big\}.
\end{split}
\end{equation}
and $\rho((T_1,T_2),(T_1',T_2')) = \sqrt{\norm{T_1-T_1'}_2^2+\norm{T_2-T_2'}_2^2}$.

First, note that $g_\alpha(T_{i-1},T_i,S) = \norm{S \circ [\alpha T_{i-1}]-T_i}_2^2$ and $g_\alpha(T'_{i-1},T'_i,S)=\norm{S \circ [\alpha T'_{i-1}]-T'_i}_2^2$. When $\alpha\geq 0$, we have $\norm{[\alpha T_{i-1}]-[\alpha T'_{i-1}]}_2\leq \alpha \norm{T_{i-1}-T'_{i-1}}_2$. When $\alpha<0$, we can use Lemma \ref{norm_inverse} to conclude that $\norm{[\alpha T_{i-1}]-[\alpha T'_{i-1}]}_2\leq \alpha \norm{T_{i-1}-T'_{i-1}}_2^b$ for some $b\geq \frac{1}{2}$.
As $S$ is Lipschitz, we can deduce that $\Delta_g(t) \leq C \alpha t^b$, for some $b>0$. Therefore the integral is finite.
\end{proof}

\begin{proof}[Proof of Theorem \ref{CLT_g}]
    From Lemma \ref{CLT_maps}, we see that $S_{N,\alpha}$ converges in probability to $S_\alpha$ and we also obtained a central limit theorem for $S_{N,\alpha}$. Then Lemma \ref{frechet} and \ref{SDC} show that $g_\alpha$ is Fr\'echet differentiable and stochastically Dini continuous, which are sufficient conditions for Corollary \ref{theorem3_modified} to be applicable, and yield a central limit theorem for $M_N(\alpha) = \frac{1}{N}\sum_{i=1}^N g_\alpha(T_{i-1},T_i,S_{N,\alpha})$ :
    $$\sqrt{N} [M_N(\alpha)-M(\alpha)]\to N(0,\sigma^2_g).$$

    Thus for any $\alpha$, $M_N(\alpha)$ converges in probability to $M(\alpha)$. By applying Corollary 3.1 from \citet{newey1991uniform} and utilizing Lemma \ref{S_g_Lipschitz}, which establishes that $g_\alpha$ satisfies Lipschitz continuity with respect to $\alpha$, we can achieve uniform convergence in probability of $M_N$ to $M$ with respect to $\alpha$:
   
    $$\sup_\alpha |M_N(\alpha)-M(\alpha)| \to 0 \quad \text{in probability}.$$
\end{proof}

\begin{proof}[Proof of Theorem \ref{consistency} (Consistency)]
Lemma \ref{CLT_g} implies that $M_N$ converges uniformly in probability to $M$ with respect to $\alpha$, and Lemma \ref{minimizer_expectation} shows that $\at$ is the unique minimizer of $M$. By applying \citet[Theorem 3.2.3]{van1996weak}, we can conclude that the estimator $\hat{\alpha}_N = \arg\min_\alpha M_N(\alpha)$ converges to the true parameter $\arg\min_\alpha M(\alpha) = \at$.
\end{proof}

We will now employ M-estimation theory to establish the convergence rate of our estimator. In order to do so, we recall the following theorem, which is taken from \citet{van1996weak}.

\begin{theorem}[\citet{van1996weak}, Theorem 3.2.5.]{\label{theorem3.2.5}}
Let $M_N$ be a stochastic process indexed by a metric space $\Theta$, and let $M$ be a deterministic function, such that for every $\theta$ in a neighborhood of $\theta_0$,
$$M(\theta)-M(\theta_0)\gtrsim d^2(\theta,\theta_0).$$
Suppose that, for every $N$ and sufficiently small $\delta$,

$$\E^* \sup_{d^2(\theta,\theta_0)<\delta} \sqrt{N}\big|(M_N-M)(\theta)-(M_N-M)(\theta_0)\big|\lesssim \phi_N(\delta),$$
for functions $\phi_N$ such that $\delta \to \phi_N(\delta)/\delta^\alpha$ is decreasing for some $\alpha<2$ (not depending on $N$). Let
$$r_N^2\phi_N\left(\frac{1}{r_N}\right)\leq \sqrt{N}, \quad \text{for every } N.$$
If the sequence $\hat{\theta}_N$ satisfies $M_N(\hat{\theta}_N)\leq M_N(\theta_0)+o_\p(r_N^{-2})$, and converges in outer probability to $\theta_0$, then $r_N d(\hat{\theta}_N,\theta_0)=O^*_\p(1)$. If the displayed conditions are valid for every $\theta$ and $\delta$, then the condition that $\hat{\theta}_N$ is consistent is unnecessary.

\end{theorem}

\begin{lemma}\label{derivative}
Let $\mathcal{T}_{l,u} = \{ T\in \mathcal{T}:  0<L_l\leq T' \leq L_u <\infty\}$ and suppose $\{T_i\} \subset \mathcal{T}_{l,u}$. Then
$$\E |M'_N(\at)| \lesssim \frac{1}{\sqrt{N}}.$$
\end{lemma}
\begin{proof}
Note that
$$M'_N(\alpha)= \frac{1}{N}\sum_{j=1}^N\frac{\partial g_\alpha(T_{j-1},T_j,S_{N,\alpha})}{\partial \alpha},$$
and
\begin{equation*}
\begin{split}
\frac{\partial g_\alpha(T_{j-1},T_j,S_{N,\alpha})}{\partial \alpha}&=\frac{\partial \norm{S_{N,\alpha} \circ [\alpha T_{j-1}]-T_j}_2^2}{\partial \alpha}\\
&=\int 2 |S_{N,\alpha} \circ [\alpha T_{j-1}](x)-T_j(x)|\times \frac{\partial}{\partial \alpha}S_{N,\alpha} \circ [\alpha T_j](x) \diff x
\end{split}
\end{equation*}
The expression $|S_{N,\alpha} \circ [\alpha T_{j-1}](x)-T_j(x)|$ can be uniformly bounded.
In what follows we will explicitly calculate $\frac{\partial}{\partial \alpha}S_{N,\alpha} \circ [\alpha T_j](x)$ for a fixed $j$.
The calculation is tedious but elementary. To calculate the derivative we use the following fact: if $f(\alpha,x) = C(\alpha, y(x,\alpha))$, then
$$\frac{\partial f}{\partial \alpha} = \frac{\partial C(\alpha,y(x,\alpha'))}{\partial \alpha}|_{\alpha'=\alpha} + \frac{\partial C(\alpha,y)}{\partial y} \times \frac{\partial y(x,\alpha)}{\partial \alpha}.$$
Using the above equation we can write:  
\begin{equation}\label{der}
\begin{split}
\frac{\partial}{\partial \alpha}S_{N,\alpha} \circ [\alpha T_j](x)&=\frac{\partial}{\partial \alpha}S_{N,\alpha} ([\alpha' T_j](x))|_{\alpha'=\alpha}+\frac{\partial S_{N,\alpha}([\alpha T_j](x))}{\partial ([\alpha T_j](x))}\times \frac{\partial [\alpha T_j](x)}{\partial \alpha}\\
& = \frac{\partial S_{N,\alpha} (y)}{\partial \alpha}|_{y = [\alpha T_j](x)} + \frac{\partial S_{N,\alpha} (y)}{\partial y}|_{y = [\alpha T_j](x)}\times \frac{\partial [\alpha T_j](x)}{\partial \alpha}.
\end{split}
\end{equation}
First, we derive the first term on the LHS of \eqref{der}:
$$ \frac{\partial S_{N,\alpha} (y)}{\partial \alpha} = \sum_{i=1}^N \frac{\partial}{\partial \alpha} T_i \circ [\alpha T_{i-1}]^{-1} (y)$$
If we consider one of the terms in this summation we have
\begin{equation}
\begin{split}
\frac{\partial}{\partial \alpha} T_i \circ [\alpha T_{i-1}]^{-1} (y) & = \frac{\partial T_i ([\alpha T_{i-1}]^{-1}(y))}{\partial [\alpha T_{i-1}]^{-1}(y)} \times \frac{\partial [\alpha T_{i-1}]^{-1}(y)}{\partial \alpha}\\
&= T'_i(z_i)|_{z_i = [\alpha T_{i-1}]^{-1}(y)} \times \frac{\partial}{\partial \alpha} [\alpha T_{i-1}]^{-1}(y)\
\end{split}
\end{equation}
Now to calculate $\frac{\partial}{\partial \alpha} [\alpha T_{i-1}]^{-1}(y)$ note that:
\begin{equation}
\begin{split}
0 = \frac{\partial}{\partial \alpha } y &= \frac{\partial}{\partial \alpha} [\alpha T_{i-1}]([\alpha T_{i-1}]^{-1}(y))\\
& = \frac{\partial}{\partial \alpha} [\alpha T_{i-1}]([\alpha' T_{i-1}]^{-1}(y))|_{\alpha'= \alpha} + \frac{\partial[\alpha T_{i-1}]([\alpha T_{i-1}]^{-1}(y))}{\partial [\alpha T_{i-1}]^{-1}(y)} \times \frac{\partial [\alpha T_{i-1}]^{-1}(y)}{\partial \alpha}\\
&= \frac{\partial}{\partial \alpha} [\alpha T_{i-1}](z_i)|_{z_i=[\alpha T_{i-1}]^{-1}(y)} +\frac{\partial [\alpha T_{i-1}](z_i)}{\partial z_i}|_{z_i = [\alpha T_{i-1}]^{-1}(y)}\times \frac{\partial [\alpha T_{i-1}]^{-1}(y)}{\partial \alpha}
\end{split}
\end{equation}
Thus
\begin{equation}
\begin{split}
\frac{\partial [\alpha T_{i-1}]^{-1}(y)}{\partial \alpha} &= (-1)\times  \frac{\partial}{\partial \alpha} [\alpha T_{i-1}](z_i)|_{z_i=[\alpha T_{i-1}]^{-1}(y)} \times \frac{1}{\frac{\partial [\alpha T_{i-1}](z_i)}{\partial z_i}|_{z_i = [\alpha T_{i-1}]^{-1}(y)} } \\
&=  \begin{cases}
(z_i-T_{i-1}(z_i))\times \frac{1}{\alpha(T'_{i-1}(z_i)-1)+1}\Big|_{z_i = [\alpha T_{i-1}]^{-1}(y)} , & \text{for } 0<\alpha\leq 1\\
(T^{-1}_{i-1}(z_i)-z_i)\times \frac{1}{\alpha(1-(T^{-1}_{i-1})'(z_i))+1}\Big|_{z_i = [\alpha T_{i-1}]^{-1}(y)}, & \text{for } -1\leq \alpha<0,
\end{cases}
\end{split}
\end{equation}
And we can conclude that
\begin{equation}
\begin{split}
\frac{\partial}{\partial \alpha} T_i \circ [\alpha T_{i-1}]^{-1} (y) &= T'_i(z_i) \times
\begin{cases}
(z_i-T_{i-1}(z_i))\times \frac{1}{\alpha(T'_{i-1}(z_i)-1)+1}\Big|_{z_i = [\alpha T_{i-1}]^{-1}(y)} , & \text{for } 0<\alpha\leq 1\\
(T^{-1}_{i-1}(z_i)-z_i)\times \frac{1}{\alpha(1-(T^{-1}_{i-1})'(z_i))+1}\Big|_{z_i = [\alpha T_{i-1}]^{-1}(y)}, & \text{for } -1\leq \alpha<0
\end{cases}
\end{split}
\end{equation}
With this, we have all the needed terms to calculate the left terms of \eqref{der}.
Now we calculate the right term of \eqref{der}:
\begin{equation}
\begin{split}
\frac{\partial S_{N,\alpha} (y)}{\partial y} &=  \frac{1}{N}\sum_{i=1}^N \frac{\partial}{\partial y} T_i \circ [\alpha T_{i-1}]^{-1}(y)\\
&= \frac{1}{N}\sum_{i=1}^N  T'_i(z_i)|_{z_i = [\alpha T_{i-1}]^{-1}(y) } \times \frac{1}{\frac{\partial [\alpha T_{i-1}](z_i)}{\partial z_i}|_{z_i = [\alpha T_{i-1}]^{-1}(y) }}\\
&=\begin{cases}
\frac{1}{N}\sum_{i=1}^N T'_i(z_i)\times \frac{1}{\alpha(T'_{i-1}(z_i)-1)+1}\Big|_{z_i = [\alpha T_{i-1}]^{-1}(y)} , & \text{for } 0<\alpha\leq 1\\
\frac{1}{N}\sum_{i=1}^N T'_i(z_i)\times \frac{1}{\alpha(1-(T^{-1}_{i-1})'(z_i))+1}\Big|_{z_i = [\alpha T_{i-1}]^{-1}(y)} , & \text{for } -1\leq \alpha<0
\end{cases}
\end{split}
\end{equation}
By plugging all the terms calculated above in \eqref{der} we get
\begin{equation}
\begin{split}
&\frac{\partial}{\partial \alpha}S_{N,\alpha} \circ [\alpha T_j](x)\\
&=\begin{cases}
\frac{1}{N}\sum_{i=1}^N T'_i(z_i)\times \frac{1}{\alpha(T'_{i-1}(z_i)-1)+1}\times(z_i-T_{i-1}(z_i)+T_j(x)-x)\Big|_{z_i = [\alpha T_{i-1}]^{-1}(y)}, & \text{for } 0<\alpha\leq 1\\
\frac{1}{N}\sum_{i=1}^N T'_i(z_i)\times\frac{1}{\alpha(1-(T^{-1}_{i-1})'(z_i))+1}\times (T^{-1}_{i-1}(z_i)-z_i+x-T^{-1}_j(x))\Big|_{z_i = [\alpha T_{i-1}]^{-1}(y)}, & \text{for } -1\leq\alpha<0,
\end{cases}
\end{split}
\end{equation}
where $y=[\alpha T_j](x).$

The differentiability of $M_N$ with respect to $\alpha$ follows from the equation above. Similarly, if we replace $S_{N,\alpha}$ with $S_\alpha$, the summations can be replaced by an integral, and we can see that $M$ is also differentiable with respect to $\alpha$. Let  $$g'(T_{j-1},T_j,S_{N,\alpha},\alpha)=\frac{\partial g_\alpha(T_{j-1},T_j,S_{N,\alpha})}{\partial \alpha},\quad g'(T_{j-1},T_j,S_\alpha,\alpha)=\frac{\partial g_\alpha(T_{j-1},T_j,S_\alpha)}{\partial \alpha}.$$
Since $\at$ is the minimizer of $M$, we must have $\E g'(T_{j-1},T_j,S_{\alpha},\alpha)|_{\alpha=\at}=M'(\at)=0$. Additionally, We can argue $g'$ is stochastically Dini-continuous when $\{T_i\} \subset \mathcal{T}_{l,u}$ (similar to the arguments in the proof of Lemma \ref{SDC}). Therefore the assumptions of Corollary \ref{theorem3_modified} (CLT) are satisfied for $g'$, and we have $\E|M_N(\at)|\lesssim \frac{1}{\sqrt{N}}$.
\end{proof}

\begin{proof}[Proof of Theorem \ref{rate} (Convergence Rate)]
Using Theorem \ref{theorem3.2.5}, we can obtain a rate of convergence for our estimator. First, it should be noted that the functional $M$ is twice differentiable with respect to $\alpha$ since it is a composition of twice differentiable functions. As $\at$ is the unique minimizer of $M$, its first derivative vanishes at $\at$, which implies that $M$ has quadratic growth around $\at$. Next, we need to find a function $\phi_N(\delta)$ such that

\begin{equation}{\label{phi_N}}
\begin{split}
\E \sup_{|\alpha-\at|\leq \delta} \sqrt{N} \Big|(M_{N}-M)(\alpha)-(M_{N}-M)(\at)\Big|&\leq \phi_N(\delta).
\end{split}
\end{equation}

Taylor expanding, we can write:
\begin{equation}
\begin{split}
(M_{N}-M)(\alpha) = (M_{N}-M)(\at) + (M_{N}-M)'(\at)\times(\alpha-\at) +\text{higher order terms}
\end{split}
\end{equation}
Since $\at$ is the minimiser of $M$, yielding $M'(\at)=0$, we only need to calculate $M'_N(\at)$. But by Lemma \ref{derivative} we can see that
$$\E |M'_N(\at)|\lesssim \frac{1}{\sqrt{N}}.$$
By plugging the inequality into the expression \eqref{phi_N} we obtain    
\begin{equation*}
\begin{split}
\E \sqrt{N}\Big|(M_{N}-M)(\alpha)-(M_{N}-M)(\at)\Big| &\leq |\alpha-\at|.
\end{split}
\end{equation*}
And, we conclude $\phi_N(\delta) =  \delta$ and the rate of convergence for $\hat{\alpha}_N$ is $N^{-\frac{1}{2}}$. Using Lemma \ref{norm_inverse} we can see
$$\norm{S_{N,\hat{\alpha}_N}-\St}_2\leq \norm{S_{N,\hat{\alpha}_N}-S_{N,\at}}_2+\norm{S_{N,\at}-\St}_2\lesssim N^{-\frac{b}{2}}+N^{-\frac{1}{2}}\lesssim N^{-\frac{b}{2}},$$
and since $\{T_i\} \subset \mathcal{T}_{l,u}$, $b=1$ according to Lemma \ref{norm_inverse}.

\end{proof}

\subsection{Generalization of Iterated System \eqref{model_muller}} \label{sec:gq-model}
The definition of the iterated system \eqref{Model} is based on the contraction of maps around the identity map. It extends system \eqref{model_muller} by introducing the map $S$.  However, we could alternatively generalise \eqref{model_muller} by introducing  $S$ not at the level of the iteration itself, but rather at the level of the contraction itself: contracting around an arbitrary map $S$, instead of the identity.  
Specifically, define the $\alpha$-contraction of a map $T$ around an arbitrary map $S$ as follows:
\begin{equation}{\label{scalar_multiplication_S}}
    \alpha [T,S](x) := \begin{cases}
        S(x)+\alpha(T(x)-S(x)) & 0<\alpha \leq 1 \\
        S(x) & \alpha=0 \\
        S(x)+\alpha(S(x)-T^{-1}(x)) & -1\leq\alpha <0.
     \end{cases}
\end{equation}
With this definition, the original contraction operation \eqref{scalar_multiplication} now corresponds to $\alpha[T,\id]$, for $\id(x)=x$ the identity map. Definition \eqref{scalar_multiplication_S} directly leads to the following extension of system \eqref{model_muller}
\begin{equation}{\label{model_muller_generalise}}
    T_i = T_{\epsilon_i} \circ \alpha [T_{i-1},S],
\end{equation}
where $\{T_{\epsilon_i}\}_{i=1}^{N}$ is again a collection of independent and identically distributed random optimal maps satisfying $\E\{T_{\epsilon_i}(x)\}=x$ almost everywhere on $\Omega$. Compared to system \eqref{Model},
$$  T_i  =T_{\epsilon_i} \circ S \circ \alpha [T_{i-1},\id].$$
this system interjects $S$ at the level of the contraction and not at the level of the random perturbation (note that for identifiability reasons it does not make sense to do both). Of course, either is more general than system \eqref{model_muller}
$$T_i  =T_{\epsilon_i}  \circ \alpha[T_{i-1},\id].$$
\begin{remark}\label{rem:gq-model}
Suppose we use the contraction definition \eqref{scalar_multiplication_S}, and define the iteration \eqref{model_muller_generalise}. Then, the quantile model (UQ) with $S = F^{-1}_\mu$ (i.e. where we contract around the quantile function of a measure $\mu$) is equivalent to the generalised quantile model (GQ) with $S = \id$; that is, they produce the same stationary time series. To demonstrate this equivalence, consider the model (GQ) with $S=\id$. We then have:
$$\E(F^{-1}_{\mu_i} \circ F_\mu (x) | F^{-1}_{\mu_{i-1}} \circ F_\mu) = \E(F^{-1}_{\mu_i} | F^{-1}_{\mu_{i-1}} \circ F_\mu) \circ F_\mu (x) = x+\alpha (F^{-1}_{\mu_{i-1}} (F_\mu(x))-x)$$
Thus,
$$\E(F^{-1}_{\mu_i} | F^{-1}_{\mu_{i-1}} \circ F_\mu) = F^{-1}_\mu (x)+\alpha (F^{-1}_{\mu_{i-1}} (x)- F^{-1}_\mu (x)),$$
which is equal to the conditional expectation of $\E(F^{-1}_{\mu_i} | F^{-1}_{\mu_{i-1}})$ when we use model \eqref{scalar_multiplication_S} for $F^{-1}_{\mu_i}$ and contract around $S = F^{-1}_\mu$.
\end{remark}

\begin{remark}
    Note that $\alpha[T,S] = T$ when $\alpha=1$, and $\alpha[T,S] = T^{-1}$ when $\alpha=-1$. Therefore in either of these cases, the time series $T_i$ does not provide any information about $S$ and it would impossible to estimate the map $S$. Therefore we assume $-1<\alpha<1$. This is in contrast with system \eqref{Model}, where consistent estimation is possible for all values of $0\leq\alpha\leq 1$,
\end{remark}
\noindent If a stationary solution to system \eqref{model_muller_generalise} exists, then
$$
        \E[T_{i}] = \E[T_{i+1}]
        = \E[E[T_{i+1}|T_{i}]]
        = \E[\alpha [T_{i},S]],
$$
and therefore $\E [T_i] = S$, when $-1<\alpha<1$.

We define the estimators $(\hat{\alpha}_N,S_N)$ of $(\alpha,S)$ as follows:
$$
    \hat{\alpha}_N \coloneqq \arg\min_\alpha M_N(\alpha),
$$
where
\begin{equation*}
\begin{split}
&M_N(\alpha)\coloneqq \frac{1}{N} \sum_{i=1}^N g (T_{i-1},T_i,S_{N})\\
&g_\alpha(T_{i-1},T_i,S) \coloneqq  \norm{\alpha [T_{i-1},S]-T_i}_2^2\\
    &S_{N} \coloneqq \frac{1}{N} \sum_{j=1}^N T_j    
\end{split}
\end{equation*}

It is worth noting that unlike in system \eqref{Model}, where the estimation of the map $S$ depends on the estimator of $\alpha$, in this system, the estimator of the map $S$ is simply the average of the maps $T_i$. Consequently, the statistical analysis of the estimators is somewhat easier in this case. Similar procedures to those used for model \eqref{Model} can be used to demonstrate the existence of a unique stationary solution, the consistency of the estimator, and obtain the rate of convergence.

Assuming that system \eqref{model_muller_generalise} satisfies the moment contracting condition \ref{moment_contracting_assumption}, a unique stationary solution for this system exists, and $\E[T_i] = S$, as in the previous case. We can then use Lemma \ref{m_dependent} to obtain the central limit theorem (CLT) for $S_N$ and show that $S_N$ converges in probability to the true $S$.

It is worth noting that the Lipschitz continuity property of the new function $g$ with respect to $\alpha$ can be shown using the fact that $\norm{\alpha_1 [T,S] - \alpha_2 [T,S]}_2\lesssim |\alpha_1-\alpha_2|$. Using this property and following a similar proof technique as in Theorem \ref{rate}, we can argue that the rate of convergence is $N^{-1/2}$.

\begin{remark}
    Once again we can use the system \eqref{model_muller_generalise} to construct a Markov chain model for a dependent sequence of probability distributions $\mu_i\in\W_2(\Omega)$ by either interpreting the maps as consecutive optimal maps between a time series of probability distributions or directly using the maps to model the quantile functions. While using system \eqref{Model}, the increment interpretation using $\alpha=0$ is equivalent to quantile interpretation using $\alpha=1$, a similar straightforward relationship does not appear to exist when using system \eqref{model_muller_generalise}.
\end{remark}

\section*{Data Availability Statement}
The data that support the findings of this study are openly available at \hfill\phantom{a} \texttt{www.ncei.noaa.gov/cdo-web/search?datasetid=GHCND}.

\bibliographystyle{imsart-nameyear}
\bibliography{complex}

\end{document}